\documentclass[twocolumn,10pt]{IEEEtran}

\usepackage[acronym]{glossaries}
\usepackage{amssymb}
\usepackage{amsmath}
\usepackage{amsthm}
\usepackage{graphicx}
\usepackage{setspace}
\usepackage{comment}
\newacronym{awgn}{AWGN}{additive white Gaussian noise}
\newacronym{bs}{BS}{base station}
\newacronym{ap}{AP}{access point}
\newacronym{ue}{UE}{user equipment}
\newacronym{csi}{CSI}{channel state information}
\newacronym{los}{LoS}{line-of-sight}
\newacronym{mimo}{MIMO}{multiple-input multiple-output}
\newacronym{miso}{MISO}{multiple-input single-output}
\newacronym{nlos}{NLoS}{non-line-of-sight}
\newacronym{snr}{SNR}{signal-to-noise ratio}
\newacronym{upa}{UPA}{uniform planar array}
\newacronym{ula}{ULA}{uniform linear array}
\newacronym{rf}{RF}{radio frequency}
\newacronym{siso}{SISO}{single-input single-output}
\newacronym{sinr}{SINR}{signal-to-interference-plus-noise ratio}
\newacronym{star-ris}{STAR-RIS}{simultaneously transmitting and reflecting RIS}
\newacronym{mrt}{MRT}{maximum ratio transmission}
\newacronym{ios}{IOS}{intelligent omni-surface}
\newacronym{pcb}{PCB}{printed circuit board}
\newacronym{ofdm}{OFDM}{orthogonal frequency division multiplexing}
\newacronym{rsma}{RSMA}{rate splitting multiple access}
\newacronym{noma}{NOMA}{non-orthogonal multiple access}
\newacronym{wpt}{WPT}{wireless power transfer}
\newacronym{swipt}{SWIPT}{simultaneous wireless information and power transfer}
\newacronym{bcd}{BCD}{block coordinate descent}
\usepackage{xcolor}
\usepackage{lscape}
\usepackage{arydshln}
\usepackage{booktabs}
\usepackage{url}

\usepackage{algorithm}
\usepackage{algorithmic}

\newtheorem{lemma}{Lemma}
\newtheorem{proposition}{Proposition}

\begin{document}

\title{Beyond Diagonal Reconfigurable Intelligent Surfaces Utilizing Graph Theory: Modeling, Architecture Design, and Optimization}

\author{Matteo~Nerini,~\IEEEmembership{Graduate Student Member,~IEEE},
        Shanpu~Shen,~\IEEEmembership{Senior Member,~IEEE},\\
        Hongyu~Li,~\IEEEmembership{Graduate Student Member,~IEEE},
        Bruno~Clerckx,~\IEEEmembership{Fellow,~IEEE}
        
\thanks{Corresponding author: Shanpu Shen.}
\thanks{M. Nerini, H. Li, and B. Clerckx are with the Department of Electrical and Electronic Engineering, Imperial College London, London SW7 2AZ, U.K. (e-mail: \{m.nerini20, c.li21, b.clerckx\}@imperial.ac.uk).}
\thanks{S. Shen is with the Department of Electrical Engineering and Electronics, University of Liverpool, Liverpool L69 3GJ, U.K. (e-mail: Shanpu.Shen@liverpool.ac.uk).}}

\maketitle

\begin{abstract}
Recently, beyond diagonal reconfigurable intelligent surface (BD-RIS) has been proposed to generalize conventional RIS.
BD-RIS has a scattering matrix that is not restricted to being diagonal and thus brings a performance improvement over conventional RIS.
While different BD-RIS architectures have been proposed, it still remains an open problem to develop a systematic approach to design BD-RIS architectures achieving the optimal trade-off between performance and circuit complexity.
In this work, we propose novel modeling, architecture design, and optimization for BD-RIS based on graph theory.
This graph theoretical modeling allows us to develop two new efficient BD-RIS architectures, denoted as tree-connected and forest-connected RIS.
Tree-connected RIS, whose corresponding graph is a tree, is proven to be the least complex BD-RIS architecture able to achieve the performance upper bound in \gls{miso} systems.
Besides, forest-connected RIS allows us to strike a balance between performance and complexity, further decreasing the complexity over tree-connected RIS.
To optimize tree-connected RIS, we derive a closed-form global optimal solution, while forest-connected RIS is optimized through a low-complexity iterative algorithm.
Numerical results confirm that tree-connected (resp. forest-connected) RIS achieves the same performance as fully-connected (resp. group-connected) RIS, while reducing the complexity by up to 16.4 times.
\end{abstract}

\glsresetall


\section{Introduction}
\label{sec:intro}

Reconfigurable intelligent surface (RIS) is expected to be a key
technology in 6G to enhance the performance  of wireless systems in
an efficient and cost-effective manner \cite{bas19,wu19,wu21}. An RIS
consists of a large number of reconfigurable scattering elements, each
with the capability to manipulate the  phase of an incident electromagnetic
wave, so that the phase shifts of these elements can be coordinated
to direct the scattered electromagnetic signal toward the intended
receiver. Owing to the benefits of low power consumption and cost-effective
architecture, RIS has gained widespread attention. 

In a conventional RIS architecture,  each element is independently
controlled by a tunable impedance connected to ground \cite{she20}.
Such architecture, referred to as the single-connected RIS, results
in a diagonal scattering matrix, which is also commonly known as the
phase shift matrix. Conventional RIS with diagonal scattering matrix
has been widely used to enhance different wireless systems, such
as single-cell \cite{wu19b,guo20},  \cite{liu21}  and multi-cell
communications  \cite{pan20}, \gls{ofdm}  \cite{li21}, \gls{noma}
\cite{xiu21},  \gls{rsma} \cite{ban21},  \gls{rf} sensing
systems \cite{hu20}, \gls{wpt} systems \cite{fen22}, and \gls{swipt}
systems \cite{zha22a}.  In addition, research has been conducted
to optimize RIS based on imperfect \gls{csi} \cite{pen22,che22},
in multi-RIS scenarios \cite{zhe21,mei22}, and with  discrete phase
shifts to account for practical hardware impairments \cite{wu19c,di20}.
Channel estimation protocols with reduced pilot overhead for RIS-aided systems have been proposed \cite{guo22}.
Active RIS have been proposed to overcome the multiplicative fading effect limiting the gains of passive RIS by amplifying the reflected signals via amplifiers integrated into their elements \cite{lon21,zha23}.
Furthermore, RIS prototypes have been realized in \cite{dai20,rao22}.

\begin{figure*}[t]
\centering{}\includegraphics[width=0.98\textwidth]{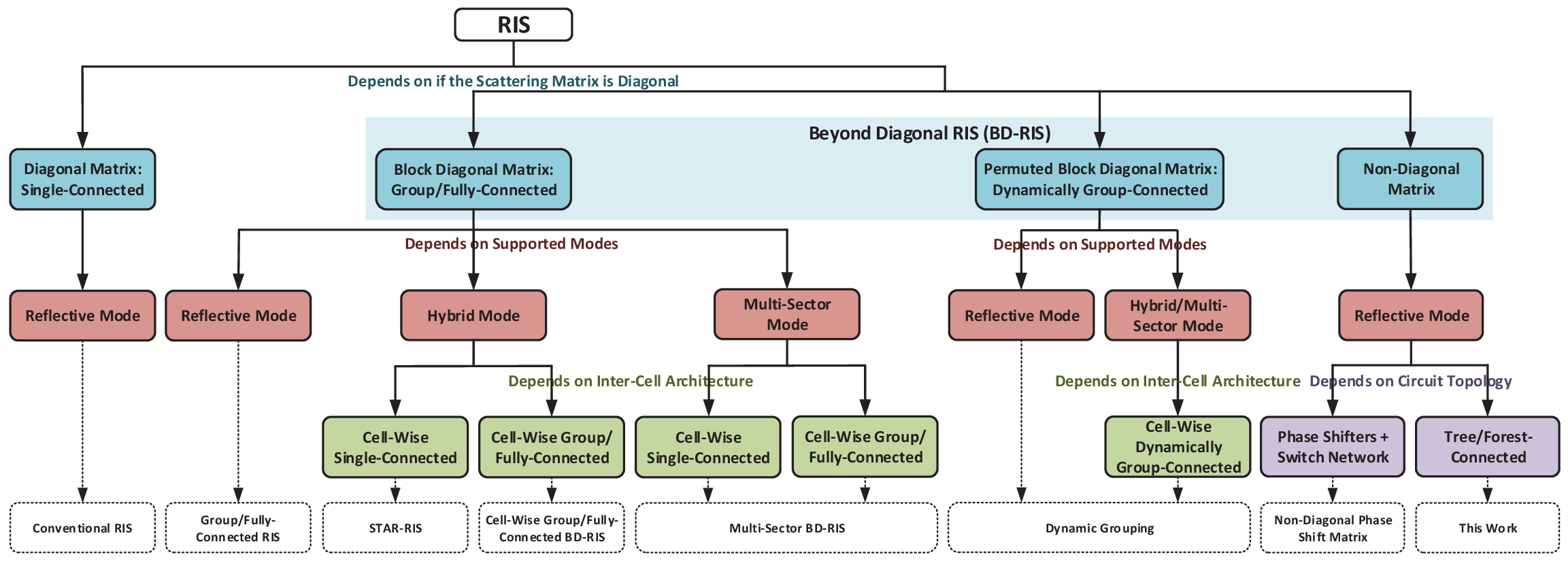}
\caption{RIS classification tree.}
\label{fig:ris-tree}
\end{figure*}

Recently, the conventional RIS architecture considered in \cite{wu19b}-\cite{rao22}
has been generalized with the introduction of beyond diagonal RIS
(BD-RIS). BD-RIS has a scattering matrix that is not restricted
to being diagonal \cite{li23-1}, with various architectures and modes
 available as depicted in the classification tree in Fig.~\ref{fig:ris-tree}.
Specifically, group-connected and fully-connected architectures, 
whereby some or all RIS elements are connected with each other through
 tunable impedance components,  are proposed in \cite{she20}, generalizing
the single-connected architecture to enhance the RIS performance.
Subsequently, group- and fully-connected RIS with discrete values
are efficiently optimized in \cite{ner21}, while the closed-form
global optimal solutions for group- and fully-connected RIS with
continuous values are derived  in \cite{ner22}. In addition,  the
\gls{star-ris} or \gls{ios} is presented in \cite{xu21,zha22b},
which differs from conventional RIS as signals impinging on this
RIS can be both reflected and transmitted through the RIS for a full
space coverage. \gls{star-ris} is generalized in \cite{li22-1}
which unifies different BD-RIS  modes (reflective/transmissive/hybrid)
and different BD-RIS architectures (single-/group-/fully-connected).
To further enhance RIS performance while maintaining a full-space
coverage, multi-sector BD-RIS is proposed in \cite{li22-2}, where
the elements are split into multiple sectors with each covering a
sector of space. Further, the synergy of multi-sector BD-RIS and
\gls{rsma} is proved to enlarge the coverage, improve the performance,
and save on antennas in multi-user systems \cite{li23-2}. Departing
from existing works on BD-RIS with fixed architectures \cite{she20,li23-1}-\cite{li23-2},
dynamically group-connected RIS has been proposed in \cite{li22-3},
 outperforming fixed group-connected RIS by dynamically dividing
the RIS elements into groups depending on the channel realization.
Besides, a BD-RIS with a non-diagonal phase shift matrix has been
proposed in \cite{li22}, which is able to achieve a higher rate than
conventional RIS.

Among all the BD-RISs \cite{she20,li23-1}-\cite{li22}, the fully-connected
architecture provides the highest flexibility and best performance,
which is however at the expense of high circuit complexity since it
has a large number of tunable impedance components. Although the group-connected
architecture can simplify the circuit complexity, it still remains
an open problem to develop a systematic approach to design BD-RIS
architecture achieving the optimal trade-off between performance and
circuit complexity, that is achieving the best performance with the
least complexity. To address this issue, in this work, we propose
novel modeling, architecture design, and optimization for BD-RIS by
utilizing graph theory. The contributions of this paper are summarized
as follows.

\textit{First}, we propose a novel modeling of BD-RIS architectures
 utilizing graph theory. To the authors\textquoteright{} best knowledge,
it is the first time that BD-RIS (or any other RIS) architectures are modeled through
graph theory. Specifically, we model each BD-RIS architecture by
a graph, whose vertices characterize the RIS ports and whose edges
characterize the tunable impedance components connecting the RIS
ports. This graph theoretical modeling lays down the foundations for
developing new efficient BD-RIS architectures.

\textit{Second}, we derive a necessary and sufficient condition for
a BD-RIS architecture to achieve the performance upper bound in \gls{miso}
systems. Consequently, we characterize the least complex BD-RISs achieving
such an upper bound. The resulting least-complexity BD-RIS architectures
are referred to as tree-connected RIS since their corresponding graph
is a tree. In addition, two specific examples of tree-connected RIS
are proposed, namely tridiagonal RIS and arrowhead RIS. Tree-connected
RIS achieves the same performance in MISO systems as fully-connected
RIS with a circuit complexity significantly reduced by $16.4$ times,
when $64$ RIS elements are considered.

\textit{Third}, we prove that there always exists one and only one global
optimal solution to optimize tree-connected RIS. To find this solution,
we derive a closed-form algorithm valid for any tree-connected architecture.
With this algorithm, tree-connected RIS achieves a performance improvement
of $51.7\%$ over single-connected RIS.
We show that the proposed algorithm is optimal in single-user as well as multi-user systems.

\textit{Fourth}, we propose forest-connected RIS as an additional BD-RIS
family to reach a trade-off between performance and circuit complexity.
Forest-connected RIS serves as a bridge between single- and tree-connected
RIS. We propose a low-complexity iterative algorithm to optimize
forest-connected RIS. Numerical results show that forest-connected
RIS achieves the same performance as group-connected RIS, while reducing
the circuit complexity by $2.4$ times.

\textit{Organization}: In Section~\ref{sec:system-model}, we introduce
the graph theoretical modeling of BD-RIS. In Section~\ref{sec:BD-RIS-Architecture-Design},
we derive tree- and forest-connected RIS architectures and provide
two examples of tree-connected RIS. In Section~\ref{sec:design-su-miso}
we provide a global optimal closed-form solution to optimize tree-connected
RIS and an iterative approach to optimize forest-connected RIS.
In Section~\ref{sec:results}, we evaluate the performance of the
proposed BD-RISs. Finally, Section~\ref{sec:conclusion}  concludes
this work.
For reproducible research, the simulation code is available at \url{https://github.com/matteonerini/bdris-utilizing-graph-theory}.

\textit{Notation}: Vectors and matrices are denoted with bold lower
and bold upper letters, respectively. Scalars are represented with
letters not in bold font. $\Re\left\{ a\right\} $, $\Im\left\{ a\right\} $,
$\left|a\right|$, $\arg\left(a\right)$, and $a^{*}$ refer to the
real part, imaginary part, modulus, phase, and the complex conjugate
of a complex scalar $a$, respectively. $\left[\mathrm{\mathbf{a}}\right]_{i}$
and $\left\Vert \mathbf{a}\right\Vert $ refer to the $i$th element
and $l_{2}$-norm of a vector $\mathrm{\mathbf{a}}$, respectively.
$\mathrm{\mathbf{A}}^{T}$, $\mathrm{\mathbf{A}}^{H}$, $\left[\mathrm{\mathbf{A}}\right]_{i,j}$,
and $\Vert\mathbf{A}\Vert$ refer to the transpose, conjugate transpose,
$\left(i,j\right)$th element, and $l_{2}$-norm of a matrix $\mathbf{A}$,
respectively. $\mathbf{A\sim B}$ means that the matrices $\mathbf{A}$
and $\mathbf{B}$ are equivalent. $\mathbb{R}$ and $\mathbb{C}$
denote real and complex number set, respectively. $j=\sqrt{-1}$ denotes
imaginary unit. $\mathbf{0}$ and $\mathbf{I}$ denote an all-zero
matrix and an identity matrix, respectively. $\mathcal{CN}\left(\mathbf{0},\mathbf{I}\right)$
denotes the distribution of a circularly symmetric complex Gaussian
random vector with mean vector $\mathbf{0}$ and covariance matrix
$\mathbf{I}$ and $\sim$ stands for ``distributed as''. diag$\left(\mathbf{a}\right)$
refers to a diagonal matrix with diagonal elements being the  vector
$\mathbf{a}$. diag$\left(\mathrm{\mathbf{A}}_{1},\ldots,\mathrm{\mathbf{A}}_{N}\right)$
refers to a block diagonal matrix with blocks being $\mathrm{\mathbf{A}}_{1},\ldots,\mathrm{\mathbf{A}}_{N}$.

\section{BD-RIS Modeling Utilizing Graph Theory}
\label{sec:system-model}

Consider a \gls{miso} system with an $M$-antenna transmitter and
a single-antenna receiver, which is aided by an $N$-element RIS.
The $N$-element RIS can be modeled as $N$ antennas  connected to
an $N$-port reconfigurable impedance network that is characterized
by the scattering matrix $\boldsymbol{\Theta}\in\mathbb{C}^{N\times N}$
\cite{she20}. We assume the direct channel between the transmitter
and receiver is blocked and thus is negligible compared to the indirect
channel provided by RIS. Therefore,  the overall channel $\mathbf{h}\in\mathbb{C}^{1\times M}$
between the transmitter and receiver can be written as 
\begin{equation}
\mathbf{h}=\mathbf{h}_{RI}\boldsymbol{\Theta}\mathbf{H}_{IT},\label{eq:H}
\end{equation}
where $\mathbf{h}_{RI}\in\mathbb{C}^{1\times N}$ and $\mathbf{H}_{IT}\in\mathbb{C}^{N\times M}$
are the channel matrices from the RIS to the receiver, and from the
transmitter to the RIS, respectively. We denote the transmit signal
as $\mathbf{x}=\mathbf{w}s$, where $\mathbf{w}\in\mathbb{C}^{M\times1}$
is the precoder satisfying $\left\Vert \mathbf{w}\right\Vert =1$
and $s\in\mathbb{C}$ is the transmit symbol with  power $P_{T}=\mathrm{E}\text{[\ensuremath{\left|s\right|^{2}}]}$.
Hence,  the received signal is given by $y=\mathbf{h}\mathbf{x}+n$,
where $n$ is the \gls{awgn}.

The $N$-port reconfigurable impedance network consists of tunable
passive impedance components. In addition to S-parameters (the scattering
matrix), the $N$-port reconfigurable impedance network can be also
characterized by Y-parameters \cite{poz11}, that is the admittance
matrix denoted as $\mathbf{Y}\in\mathbb{C}^{N\times N}$. According
to microwave network theory \cite{poz11}, the scattering matrix $\boldsymbol{\Theta}$
and admittance matrix $\mathbf{Y}$ are related by 
\begin{equation}
\boldsymbol{\Theta}=\left(\mathbf{I}+Z_{0}\mathbf{Y}\right)^{-1}\left(\mathbf{I}-Z_{0}\mathbf{Y}\right),\label{eq:T(Y)}
\end{equation}
where $Z_{0}$ denotes the reference impedance used for computing
the scattering parameter, and usually set as $Z_{0}=50\:\Omega$.
Furthermore, we have $\mathbf{Y}=\mathbf{Y}^{T}$ and $\boldsymbol{\Theta}=\boldsymbol{\Theta}^{T}$
because of the reciprocity of the reconfigurable impedance network.
To maximize the power scattered by the RIS, $\mathbf{Y}$ should be
purely susceptive and thus writes as $\mathbf{Y}=j\mathbf{B}$, where
$\mathbf{B}\in\mathbb{R}^{N\times N}$ denotes the susceptance matrix
of the $N$-port reconfigurable impedance network\footnote{In this study, we focus on passive RIS, assumed to be lossless to maximize the reflection gain.
However, active RIS with negative resistive components can be considered to further improve the gain \cite{lon21,zha23}.}.
Hence, $\boldsymbol{\Theta}$
is given by 
\begin{equation}
\boldsymbol{\Theta}=\left(\mathbf{I}+jZ_{0}\mathbf{B}\right)^{-1}\left(\mathbf{I}-jZ_{0}\mathbf{B}\right),\label{eq:T(B)}
\end{equation}
yielding that $\boldsymbol{\Theta}$ is symmetric unitary in the case of reciprocal and lossless RIS.
Depending on the circuit topology of the $N$-port reconfigurable
impedance network, the scattering matrix $\boldsymbol{\Theta}$ and
admittance matrix $\mathbf{Y}$ satisfy different constraints, which
results in diagonal RIS and beyond diagonal RIS as follows. 

\subsection{Diagonal RIS}

In diagonal RIS, each RIS port is not connected to the other ports
and is connected to ground through a tunable admittance \cite{she20}.
Accordingly, the admittance matrix $\mathbf{Y}$ and the scattering
matrix $\boldsymbol{\Theta}$ are diagonal, written as
\begin{equation}
\mathbf{Y}=\mathrm{diag}\left(Y_{1},Y_{2},\ldots,Y_{N}\right),
\end{equation}
\begin{equation}
\mathbf{\boldsymbol{\Theta}}=\mathrm{diag}\left(e^{j\theta_{1}},e^{j\theta_{2}},\ldots,e^{j\theta_{N}}\right),
\end{equation}
where $Y_{n}$ and $\theta_{n}$ are respectively the tunable admittance
and phase shift for the $n$th RIS port for $n=1,\ldots,N$. Such
circuit topology is also referred to as single-connected RIS \cite{she20}
and has been widely adopted \cite{wu19b}-\cite{rao22}.

\subsection{Beyond Diagonal RIS}

In BD-RIS, the RIS ports can also be connected to each other through
additional tunable admittance components. We denote the tunable admittance
connecting the $n$th port to the $m$th port as $Y_{n,m}$. According
to \cite{she16}, given the admittance components $Y_{n}$ and $Y_{n,m}$,
the $(n,m)$th entry of the admittance matrix $\mathbf{Y}$ is given
by 
\begin{equation}
\left[\mathbf{Y}\right]_{n,m}=\begin{cases}
-Y_{n,m} & n\neq m\\
Y_{n}+\sum_{k\neq n}Y_{n,k} & n=m
\end{cases}.\label{eq:Yij}
\end{equation}
Thus, by selecting $Y_{n,m}=-\left[\mathbf{Y}\right]_{n,m}$ and $Y_{n}=\sum_{k}\left[\mathbf{Y}\right]_{n,k}$,
we can implement an arbitrary symmetric $\mathbf{Y}$ so that the
admittance matrix $\mathbf{Y}$ and the scattering matrix $\boldsymbol{\Theta}$
are beyond diagonal. Particularly, there are two special categories
of BD-RIS, namely fully-connected and group-connected RIS \cite{she20}.
Due to the enhanced flexibility in $\mathbf{Y}$ and $\boldsymbol{\Theta}$,
BD-RIS is proven to outperform single-connected RIS \cite{she20,li23-1}-\cite{li22}.

\subsection{Graph Theoretical Modeling}

To model the general circuit topology of BD-RIS, we
resort for the first time to graph theoretical tools \cite{bon76}.
Specifically, we represent the circuit topology of BD-RIS through
a graph
\begin{equation}
\mathcal{G}=\left(\mathcal{V},\mathcal{E}\right),
\end{equation}
where $\mathcal{V}$ represents the \textit{vertex set} of $\mathcal{G}$
and is given by the set of indices of RIS ports, i.e. 
\begin{equation}
\mathcal{V}=\left\{ 1,2,\ldots,N\right\} ,
\end{equation}
and $\mathcal{E}$ represents the \textit{edge set} of $\mathcal{G}$
and is given by 
\begin{equation}
\mathcal{E}=\left\{ \left(n_{\ell},m_{\ell}\right)|\:n_{\ell},m_{\ell}\in\mathcal{V},\:Y_{n_{\ell},m_{\ell}}\neq0,\:n_{\ell}\neq m_{\ell}\right\} .
\end{equation}
Thus, there exists an edge between vertex $n_{\ell}$ and vertex $m_{\ell}$
if and only if there is a tunable admittance connecting port $n_{\ell}$
and port $m_{\ell}$, namely the $\ell$th admittance. Accordingly,
given a graph $\mathcal{G}$ with edge set $\mathcal{E}=\left\{ \left(n_{\ell},m_{\ell}\right)\right\} _{\ell=1}^{L}$,
the corresponding BD-RIS admittance matrix has $2L$ non-zero off-diagonal
entries, i.e.  $\left[\mathbf{Y}\right]{}_{n_{\ell},m_{\ell}}$ and
$\left[\mathbf{Y}\right]{}_{m_{\ell},n_{\ell}}$ for $\ell=1,\ldots,L$,
while the other off-diagonal entries are zero.
In the following, we define the circuit complexity of a BD-RIS architecture as the number of tunable admittance components in its circuit topology.
Thus, the circuit complexity of a BD-RIS represented by a graph with $L$ edges is given by $N+L$ since it includes $N$ tunable admittance components connecting each element to ground, and $L$ tunable admittance components interconnecting the elements to each other.

\begin{figure}[t]
\centering{}\includegraphics[width=0.46\textwidth]{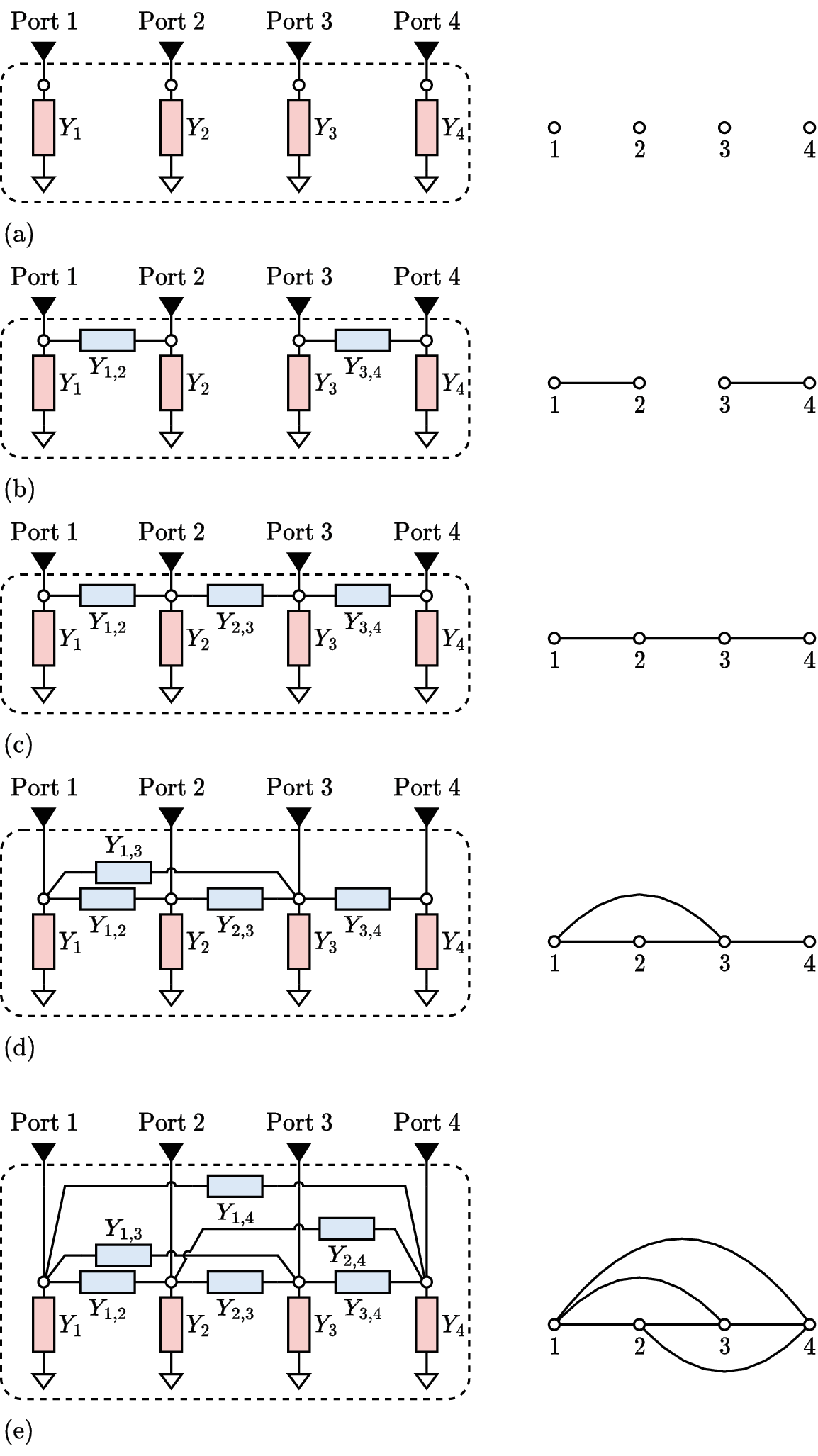}
\caption{Examples of $4$-port BD-RIS architectures (left) and their corresponding
graphs $\mathcal{G}$ (right), where $\mathcal{G}$ is (a) empty,
(b) disconnected and acyclic, (c) connected and acyclic, i.e., a tree,
 (d) connected and cyclic, (e) complete. }
\label{fig:example} 
\end{figure}

To better clarify our graph theoretical modeling of BD-RIS, we show
a $4$-element single-connected RIS with its associated graph in Fig.~\ref{fig:example}(a),
and four examples of $4$-element BD-RISs with their associated graphs
in Fig.~\ref{fig:example}(b)-(e). We also use Fig.~\ref{fig:example}
to illustrate the graph theoretical definitions \cite{bon76} that
are used in this work, as briefly summarized in the following. 
\begin{itemize}
\item A graph is\textit{ empty }when there is no edge.
\item A graph is  \textit{connected} when there is a \textit{path}, i.e.,
a finite sequence of distinct edges joining a sequence of distinct
vertices, from any vertex to any other vertex. 
\item A graph that is not connected is called \textit{disconnected}. 
\item A \textit{cyclic} graph is a graph containing at least one \textit{cycle},
i.e., a finite sequence of distinct edges joining a sequence of vertices,
where only the first and last vertices are equal. 
\item A graph that is not cyclic is  called \textit{acyclic}, or a \textit{forest}. 
\item A \textit{tree} is defined as a connected and acyclic graph. 
\item A graph is \textit{complete }when each pair of distinct vertices is
joined by an edge.
\item A graph is \textit{simple }when there is no edge with identical ends
(called \textit{loops}) and no multiple edges joining the same pair
of vertices.
\end{itemize}
Given these definitions, we recognize that the graph for the single-connected
RIS in Fig.~\ref{fig:example}(a) is empty, the graph in Fig.~\ref{fig:example}(b)
is disconnected and acyclic, the graph in Fig.~\ref{fig:example}(c)
is connected and acyclic, i.e., a tree, the graph in Fig.~\ref{fig:example}(d)
is connected and cyclic as it presents the cycle $1-2-3$, and the
graph for the fully-connected RIS in Fig.~\ref{fig:example}(e) is
complete. In addition, graphs associated with BD-RIS are always simple\footnote{In this work, we consider graphs where the edges are not oriented since we assume the RIS to have a reciprocal reconfigurable impedance network. BD-RIS architectures with non-reciprocal reconfigurable impedance networks could be modeled through \textit{directed} graphs, i.e., graphs with oriented edges, but are beyond the scope of this work.}.
Therefore, graph theory is effective to characterize and model the
circuit topology of BD-RIS and the corresponding admittance matrix
$\mathbf{Y}$.

Utilizing graph theory, we propose two novel  families of BD-RIS
with remarkable properties, which are referred to as tree-connected
RIS and forest-connected RIS, in the following section.

\section{BD-RIS Architecture Design Utilizing Graph Theory}
\label{sec:BD-RIS-Architecture-Design}

\subsection{Tree-Connected RIS}

As a starting point, we investigate how the graph of BD-RIS influences
the system performance. To that end, we formulate the BD-RIS optimization
problem for MISO systems utilizing graph theory. Specifically, we
assume the BD-RIS is lossless and purely susceptive and aim to maximize
the received signal power $P_{R}=P_{T}\left\vert \mathbf{h}_{RI}\boldsymbol{\Theta}\mathbf{H}_{IT}\mathbf{w}\right\vert ^{2}$
by jointly optimizing $\mathbf{w}$ and $\boldsymbol{\Theta}$\footnote{Maximizing the received signal power or the achievable rate are equivalent problems in single-user systems.
Thus, we consider the received signal power since it is independent of the noise power.}. Thus,
we can formulate the problem as
\begin{align}
\underset{\boldsymbol{\Theta},\mathbf{w}}{\mathsf{\mathrm{max}}}\;\; & P_{T}\left\vert \mathbf{h}_{RI}\boldsymbol{\Theta}\mathbf{H}_{IT}\mathbf{w}\right\vert ^{2}\label{eq:received power}\\
\mathsf{\mathrm{s.t.}}\;\;\; & \boldsymbol{\Theta}=\left(\mathbf{I}+jZ_{0}\mathbf{B}\right)^{-1}\left(\mathbf{I}-jZ_{0}\mathbf{B}\right),\label{eq:received power C1}\\
 & \mathbf{B}=\mathbf{B}^{T},\label{eq:received power C2}\\
 & \mathbf{B}\in\mathcal{B}_{\mathcal{G}},\label{eq:received power C3}\\
 & \Vert\mathbf{w}\Vert=1,\label{eq:received power C4}
\end{align}
where $\mathcal{B}_{\mathcal{G}}$ represents the set of possible
susceptance matrices of BD-RISs characterized by the graph $\mathcal{G}=\left(\mathcal{V},\mathcal{E}\right)$
and is written as
\begin{equation}
\mathcal{B}_{\mathcal{G}}=\left\{ \mathbf{B}|\left[\mathbf{B}\right]{}_{n,m}=0,n\neq m,n,m\in\mathcal{V},\left(n,m\right)\notin\mathcal{E}\right\} 
\end{equation}
which means that an off-diagonal entry is zero if there exists no
corresponding edge in the graph. 
The graph $\mathcal{G}$ is fixed in \eqref{eq:received power}-\eqref{eq:received power C4} since the interconnections between the RIS elements are fixed to avoid additional circuit complexity.
Note that problem \eqref{eq:received power}-\eqref{eq:received power C4} considers the optimization of the RIS based on the instantaneous channel, as typically performed in related literature \cite{bas19,wu19,wu21}.
Perfect \gls{csi} is assumed, which can be obtained with the semi-passive channel estimation protocol discussed in \cite{wu21}.

From \eqref{eq:received power}-\eqref{eq:received power C4}, we
can deduce that the maximum received signal power depends on the graph
of the BD-RIS. Specifically, when the graph $\mathcal{G}$ has more
edges, the susceptance matrix $\mathbf{B}$ is more flexible so that
the maximum received signal power is higher. However, the enhanced
performance is achieved at the expense of higher circuit complexity
(more connections from the increased number of edges). Thus, it is worth designing the graph of BD-RIS to achieve the best performance with the least
complexity. 

To that end, we first consider the case that the graph $\mathcal{G}$
is complete, that is the fully-connected RIS.
As shown in \cite{she20}
and \cite{ner22}, maximum ratio transmission (MRT) gives the optimal
$\mathbf{w}$ and subsequently the fully-connected RIS can achieve
the upper bound of the received signal power given by 
\begin{equation}
\bar{P}_{R}=P_{T}\left\Vert \mathbf{h}_{RI}\right\Vert ^{2}\left\Vert \mathbf{H}_{IT}\right\Vert ^{2},\label{eq:PR-UB}
\end{equation}
following the sub-multiplicativity of the $l_2$-norm and that $\boldsymbol{\Theta}^H\boldsymbol{\Theta}=\mathbf{I}$. Leveraging
this upper bound, herein we define a BD-RIS that  achieves the performance
upper bound $\bar{P}_{R}$ for any channel realization as MISO optimal.
Thus, the fully-connected RIS is MISO optimal. More generally, the
following lemma provides a sufficient and necessary  condition for
a BD-RIS to be MISO optimal.

\begin{lemma} A BD-RIS with associated graph $\mathcal{G}$ is MISO
optimal if and only if $\mathcal{G}$ is a connected graph. \label{lem:connected}
\end{lemma} \begin{proof} Please refer to Appendix~A. \end{proof} 

Using Lemma~\ref{lem:connected}, we have  the following proposition
providing the least complexity for MISO optimal BD-RIS.

\begin{proposition} The minimum number of edges in the graph of a
MISO optimal BD-RIS is $N-1$, with $N$ denoting the number of RIS
elements. \label{pro:N-1} \end{proposition} \begin{proof} Please
refer to Appendix~B. \end{proof}

\begin{figure}[t]
\centering{}\includegraphics[width=0.48\textwidth]{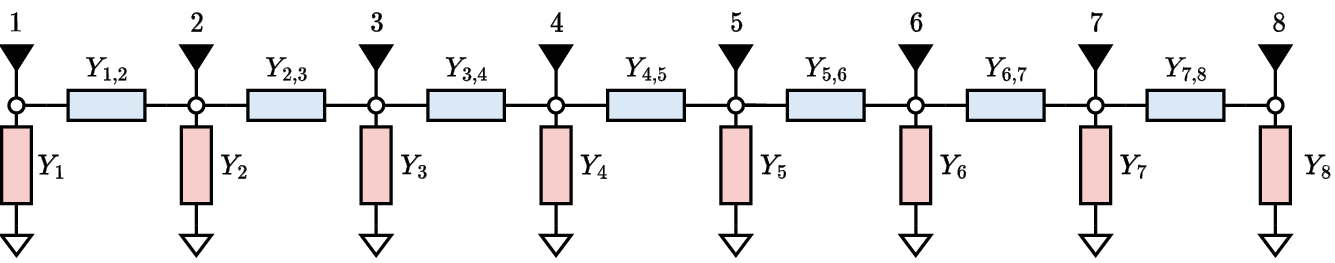}
\caption{Tridiagonal RIS with $N=8$.}
\label{fig:3diag} 
\end{figure}

Remarkably, a connected graph $\mathcal{G}$ on $N$ vertices with
$N-1$ edges is a tree. To show this, we observe that removing any
edge from $\mathcal{G}$ makes the graph disconnected because of Proposition~\ref{pro:N-1}.
Thus, $\mathcal{G}$ cannot have a cycle and must be a tree (connected
and acyclic). Therefore, to minimize the BD-RIS circuit complexity
while maintaining optimal performance, we propose tree-connected
RIS as BD-RIS whose corresponding graph is a tree. According to
this definition, tree-connected RIS includes $N$ admittance components
connecting each port to ground and $N-1$ admittance components interconnecting
the ports to each other, yielding a total of $2N-1$ admittance components,
which is much less than the $N(N+1)/2$ admittance components in fully-connected
RIS.

\subsection{Two Examples of Tree-Connected RIS}
\label{sec:examples}

Tree-connected RIS refers to a family of BD-RISs, including multiple
possible architectures. To clarify the concept of tree-connected RIS,
in this subsection we provide two specific examples of tree-connected
RIS, referred to as tridiagonal and arrowhead RIS. 

\subsubsection{Tridiagonal RIS}

We define tridiagonal RIS as tree-connected RIS whose graph $\mathcal{G}$
is a \textit{path graph.} In graph theory, the degree of a vertex
is defined as the number of edges incident with it, and a path graph
is defined as a graph where two vertices have \textit{degree} one
and all the others (if any) have degree two \cite{bon76}. The vertex
set of the graph associated with tridiagonal RIS is defined as
\begin{equation}
\mathcal{V}^{\mathrm{Tri}}=\left\{ n_{1},n_{2},\ldots,n_{N}\right\} ,\label{eq:Vtri}
\end{equation}
and the edge set can be expressed accordingly as
\begin{equation}
\mathcal{E}^{\mathrm{Tri}}=\left\{ \left(n_{\ell},n_{\ell+1}\right)|\:\ell=1,\ldots,N-1\right\} .\label{eq:Etri}
\end{equation}
Interestingly, the tridiagonal architectures have also been proposed
to realize low-complexity impedance matching networks for multiple
antenna systems \cite{she16}. In Fig.~\ref{fig:3diag}, we provide
an illustrative example of tridiagonal RIS with $N=8$. As a consequence
of such architecture, the susceptance matrix of tridiagonal RIS 
is symmetric and tridiagonal  ($[\mathbf{B}]_{i,j}=0$ if $\vert i-j\vert>1$),
written as
\begin{equation}
\mathbf{B}=\left[\begin{matrix}\left[\mathbf{B}\right]_{1,1} & \left[\mathbf{B}\right]_{1,2} & \cdots & 0\\
\left[\mathbf{B}\right]_{1,2} & \left[\mathbf{B}\right]_{2,2} & \ddots & \vdots\\
\vdots & \ddots & \ddots & \left[\mathbf{B}\right]_{N-1,N}\\
0 & \cdots & \left[\mathbf{B}\right]_{N-1,N} & \left[\mathbf{B}\right]_{N,N}
\end{matrix}\right].\label{eq:B-TRI}
\end{equation}
Note that if the port indexes in $\mathcal{V}^{\mathrm{Tri}}$ are
permuted, $\mathbf{B}$ is a permuted tridiagonal matrix.

\subsubsection{Arrowhead RIS}

\begin{figure}[t]
\centering{}\includegraphics[width=0.18\textwidth]{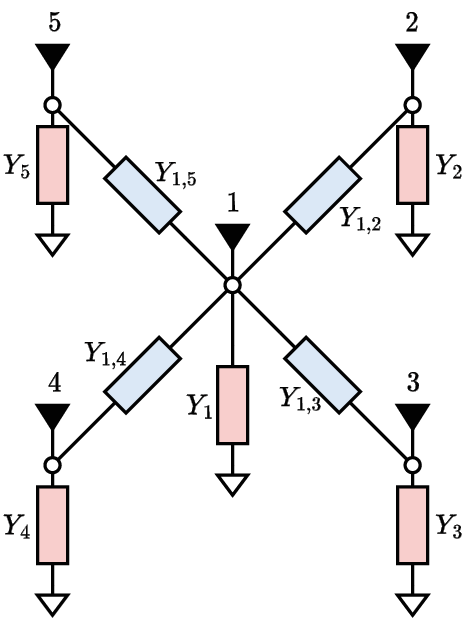}
\caption{Arrowhead RIS  with $N=5$ where port 1 is the central vertex.}
\label{fig:arrow} 
\end{figure}

We define arrowhead RIS as tree-connected RIS whose graph $\mathcal{G}$
is a \textit{star graph}. In graph theory, a star graph is defined
as a graph where at most one vertex, called \textit{central}, has
degree greater than one and all the others have degree one \cite{bon76}.
Accordingly, the vertex set of the graph associated with arrowhead
RIS is defined as
\begin{equation}
\mathcal{V}^{\mathrm{Arrow}}=\left\{ c,n_{1},n_{2},\ldots,n_{N-1}\right\} ,\label{eq:Varrow}
\end{equation}
where $c$ denotes the central vertex, while the edge set is consequently
expressed as
\begin{equation}
\mathcal{E}^{\mathrm{Arrow}}=\left\{ \left(c,n_{\ell}\right)|\:\ell=1,\ldots,N-1\right\} .\label{eq:Earrow}
\end{equation}
In Fig.~\ref{fig:arrow}, we provide an illustrative example of arrowhead
RIS with $N=5$ where port 1 is the central vertex. As a consequence,
the susceptance matrix of arrowhead RIS is symmetric and arrowhead
($[\mathbf{B}]_{i,j}=0$ if $i\neq c$, $j\neq c$, and $i\neq j$)
\cite{ole90}.  Assuming port 1 is the central vertex, the susceptance
matrix of arrowhead RIS can be written as 
\begin{equation}
\mathbf{B}=\left[\begin{matrix}\left[\mathbf{B}\right]_{1,1} & \left[\mathbf{B}\right]_{1,2} & \cdots & \left[\mathbf{B}\right]_{1,N}\\
\left[\mathbf{B}\right]_{1,2} & \left[\mathbf{B}\right]_{2,2} & \cdots & 0\\
\vdots & \vdots & \ddots & \vdots\\
\left[\mathbf{B}\right]_{1,N} & 0 & \cdots & \left[\mathbf{B}\right]_{N,N}
\end{matrix}\right].\label{eq:B-ARROW}
\end{equation}
Note that if the central vertex is not port 1, $\mathbf{B}$ is a
permuted arrowhead matrix.

It should be noted that all tree-connected RISs satisfy the necessary
condition to be MISO optimal and contain the same number of tunable
admittance components. Thus, tridiagonal and arrowhead RIS have the
same circuit complexity and MISO optimality. In practice, we can select
the suitable tree-connected RIS architecture according to the use
case.
On the one hand, tridiagonal RIS can be easily implemented with cables or tapes,
thanks to the path graph circuit topology, so it enables the
practical development of tree-connected RIS with \glspl{ula} or
radio stripes \cite{sha21}.
In addition, tridiagonal RIS is the tree-connected architecture that minimizes the length of the interconnections since it connects only adjacent RIS elements.
On the other hand, arrowhead BD-RIS is more suitable to develop tree-connected RIS with \glspl{upa}. 

\subsection{Forest-Connected RIS}

To further reduce the circuit complexity of tree-connected RIS, especially
for the large-scale case, we propose forest-connected RIS. A forest-connected
RIS is defined as a BD-RIS where the $N$ RIS elements are divided
into $G=N/N_{G}$ groups with group size $N_{G}$ and each group utilizes
the tree-connected architecture (i.e. the graph for each group is
a tree on $N_{G}$ vertices)\footnote{Note that the tree-connected RIS is a special case of the forest-connected RIS, i.e., with $G=1$. This is consistent with the fact that in graph theory a tree is a special case of forest.}. Therefore, the graph $\mathcal{G}$
associated with the forest-connected RIS is a forest, and the corresponding
vertex and edge sets can be represented as 
\begin{equation}
\mathcal{V}=\cup_{g=1}^{G}\mathcal{V}_{g},\:\mathcal{E}=\cup_{g=1}^{G}\mathcal{E}_{g},
\end{equation}
 where $\mathcal{V}_{g}$ and $\mathcal{E}_{g}$ denote the vertex
and edge sets of the graph associated with the $g$th group for $g=1,\ldots,G$,
with $\mathcal{V}_{g_{1}}\cap\mathcal{V}_{g_{2}}=\emptyset$ for all
$g_{1}\neq g_{2}$.

More specifically, the tree-connected architecture for each group
in the forest-connected RIS can be either tridiagonal or arrowhead.
When the tridiagonal architecture is used, the vertex and edge sets
of the graph associated with forest-connected RIS are defined as
\begin{equation}
\mathcal{V}_{\textrm{Forest}}^{\textrm{Tri}}=\cup_{g=1}^{G}\left\{ n_{g,1},n_{g,2},\ldots,n_{g,N_{G}}\right\} ,\label{eq:VFOREST-TRI}
\end{equation}
\begin{equation}
\mathcal{E}_{\textrm{Forest}}^{\textrm{Tri}}=\cup_{g=1}^{G}\left\{ \left(n_{g,\ell},n_{g,\ell+1}\right)|\:\ell=1,\ldots,N_{G}-1\right\} ,\label{eq:EFOREST-TRI}
\end{equation}
similarly to \eqref{eq:Vtri} and \eqref{eq:Etri}. Accordingly, the
susceptance matrix of forest-connected RIS with tridiagonal architecture
for each group is block diagonal with each block being symmetric and
tridiagonal. On the other hand, when the arrowhead architecture is
used, the vertex and edge sets of the graph associated with forest-connected
RIS are defined as
\begin{equation}
\mathcal{V}_{\textrm{Forest}}^{\textrm{Arrow}}=\cup_{g=1}^{G}\left\{ c_{g},n_{g,1},n_{g,2},\ldots,n_{g,N_{G}-1}\right\} ,\label{eq:VFOREST-ARR}
\end{equation}
\begin{equation}
\mathcal{E}_{\textrm{Forest}}^{\textrm{Arrow}}=\cup_{g=1}^{G}\left\{ \left(c_{g},n_{g,\ell}\right)|\:\ell=1,\ldots,N_{G}-1\right\} ,\label{eq:EFOREST-ARR}
\end{equation}
where $c_{g}$ is the index of the central vertex for the $g$th group,
similarly to \eqref{eq:Varrow} and \eqref{eq:Earrow}. Accordingly,
the susceptance matrix of forest-connected RIS with arrowhead architecture
for each group is block diagonal with each block being symmetric and
arrowhead.

The forest-connected RIS achieves a good performance-complexity trade-off
between the single-connected and the tree-connected RIS. Note that
the single-connected and tree-connected RIS can be viewed as two special
cases of the forest-connected RIS, with $N_{G}=1$ and $N_{G}=N$,
respectively. In addition, forest-connected RIS is equivalent to
group-connected RIS when $N_{G}=2$ \cite{she20}. The number of
tunable admittance components in forest-connected RIS with group
size $N_{G}$ is $N(2-1/N_{G})$. 

\section{Optimization of Tree- and Forest-Connected RIS for MISO Systems}
\label{sec:design-su-miso}

In this section, we optimize tree-connected and forest-connected RIS to
maximize the received signal power in \gls{miso} systems. For tree-connected
RIS, we show that any tree-connected RIS is MISO optimal and provide
a closed-form global optimal solution. For forest-connected RIS,
we provide a low-complexity iterative solution.

\subsection{Tree-Connected  RIS Optimization}

The tree-connected RIS optimization to maximize the received signal
power in MISO systems can be formulated as \eqref{eq:received power}-\eqref{eq:received power C4}
plus a constraint that graph $\mathcal{G}$ is a tree. Recall that
the received signal power in BD-RIS-aided MISO systems is upper-bounded
by $\bar{P}_{R}=P_{T}\left\Vert \mathbf{h}_{RI}\right\Vert ^{2}\left\Vert \mathbf{H}_{IT}\right\Vert ^{2}$.
In the following, we prove that there is always one and only one
solution for the optimal susceptance matrix $\mathbf{B}$ of tree-connected
RIS to achieve the upper bound $\bar{P}_{R}$, and propose an algorithm
to find such optimal $\mathbf{B}$.

Considering a RIS implemented through a reciprocal and lossless circuit, the key to achieving the upper bound $\bar{P}_{R}$ given by \eqref{eq:PR-UB}
is to find a symmetric unitary matrix $\boldsymbol{\Theta}$
satisfying 
\begin{equation}
\hat{\mathbf{h}}_{RI}^{H}=\boldsymbol{\Theta}\mathbf{u}_{IT},\label{eq:optimal fully}
\end{equation}
where $\hat{\mathbf{h}}_{RI}=\mathbf{h}_{RI}/\left\Vert \mathbf{h}_{RI}\right\Vert $
and $\mathbf{u}_{IT}=\mathbf{u}_{\text{max}}(\mathbf{H}_{IT})$ is
the dominant left singular vector of $\mathbf{H}_{IT}$\footnote{Condition \eqref{eq:optimal fully} is derived by considering
$P_{R}\leq P_{T}\underset{\left\|\mathbf{x}\right\|=1}{\mathsf{\mathrm{max}}}\;\left\|\mathbf{h}_{RI}\right\|^{2}\left\|\boldsymbol{\Theta}\mathbf{H}_{IT}\mathbf{x}\right\|^{2}\leq P_{T}\underset{\left\|\mathbf{x}\right\|=1}{\mathsf{\mathrm{max}}}\;\left\|\mathbf{h}_{RI}\right\|^{2}\left\|\boldsymbol{\Theta}\mathbf{H}_{IT}\right\|^{2}\left\|\mathbf{x}\right\|^{2}$.
Note that the equality holds in the two inequalities when $\hat{\mathbf{h}}_{RI}^H$ is equal to the dominant left singular vector of $\boldsymbol{\Theta}\mathbf{H}_{IT}$, i.e., $\boldsymbol{\Theta}\mathbf{u}_{IT}$.}.
Substituting \eqref{eq:T(B)} into \eqref{eq:optimal fully}, we can
equivalently rewrite \eqref{eq:optimal fully} as 
\begin{equation}
\left(\mathbf{I}+jZ_{0}\mathbf{B}\right)\hat{\mathbf{h}}_{RI}^{H}=\left(\mathbf{I}-jZ_{0}\mathbf{B}\right)\mathbf{u}_{IT},
\end{equation}
which  can be expressed in a compact form as 
\begin{equation}
\mathbf{B}\boldsymbol{\alpha}=\boldsymbol{\beta},\label{eq:system1}
\end{equation}
where $\boldsymbol{\alpha}=jZ_{0}\left(\mathbf{u}_{IT}+\hat{\mathbf{h}}_{RI}^{H}\right)\in\mathbb{C}^{N\times1}$
and $\boldsymbol{\beta}=\mathbf{u}_{IT}-\hat{\mathbf{h}}_{RI}^{H}\in\mathbb{C}^{N\times1}$.
The equation in \eqref{eq:system1} is composed of $N$ linear equations
with $2N-1$ unknowns. The linear equation coefficients $\boldsymbol{\alpha}$
and $\boldsymbol{\beta}$ are complex vectors, while the $2N-1$
unknowns are real since they are the entries of the susceptance matrix
$\mathbf{B}$ of the tree-connected RIS, which are not constrained to
be zero as shown in \eqref{eq:B-TRI} and \eqref{eq:B-ARROW}. We
aim to  solve \eqref{eq:system1} as it is a  sufficient and necessary
condition to achieve the performance upper bound $\bar{P}_{R}$. To
that end, we rewrite \eqref{eq:system1} as  $2N$ equations
with real coefficients and $2N-1$ real unknowns, i.e., 
\begin{equation}
\left[\begin{array}{c}
\mathbf{B}\\
\hdashline\mathbf{B}
\end{array}\right]\mathbf{a}=\mathbf{b},\label{eq:system2}
\end{equation}
where $\mathbf{a}\in\mathbb{R}^{2N\times1}$ and $\mathbf{b}\in\mathbb{R}^{2N\times1}$
are defined as 
\begin{equation}
\mathbf{a}=\left[\begin{array}{c}
\Re\left\{ \boldsymbol{\alpha}\right\} \\
\hdashline\Im\left\{ \boldsymbol{\alpha}\right\} 
\end{array}\right],\:\mathbf{b}=\left[\begin{array}{c}
\Re\left\{ \boldsymbol{\beta}\right\} \\
\hdashline\Im\left\{ \boldsymbol{\beta}\right\} 
\end{array}\right].\label{eq:b}
\end{equation}
To solve \eqref{eq:system2}, we further rewrite the equation such
that the $2N-1$ real unknowns are explicitly collected in a vector
$\mathbf{x}\in\mathbb{R}^{(2N-1)\times1}$. Specifically, given a
tree-connected RIS whose graph has an edge set $\mathcal{E}=\left\{ \left(n_{\ell},m_{\ell}\right)\right\} _{\ell=1}^{N-1}$
such as \eqref{eq:Etri} or \eqref{eq:Earrow}, the $2N-1$ real unknowns
are collected as
\begin{equation}
\mathbf{x}=\left[\left[\mathbf{B}\right]_{1,1},\ldots,\left[\mathbf{B}\right]_{N,N},\left[\mathbf{B}\right]_{n_{1},m_{1}},\ldots,\left[\mathbf{B}\right]_{n_{N-1},m_{N-1}}\right]^{T},\label{eq:x}
\end{equation}
where $\text{[\ensuremath{\mathbf{B}}]}_{n,n}$ for $n=1,\ldots,N$
denote the $N$ diagonal entries of $\mathbf{B}$ and $[\mathbf{B}]_{n_{\ell},m_{\ell}}$
for $\ell=1,\ldots,N-1$ denote the $N-1$ off-diagonal entries of
$\mathbf{B}$ specified by the edge set of graph. Accordingly, \eqref{eq:system2}
can be equivalently rewritten as 
\begin{equation}
\mathbf{A}\mathbf{x}=\mathbf{b},\label{eq:system3}
\end{equation}
 where $\mathbf{A}\in\mathbb{R}^{2N\times(2N-1)}$ is the coefficient
matrix given by 
\begin{equation}
\mathbf{A}=\left[\begin{array}{cc}
\Re\left\{ \mathbf{A}_{1}\right\}  & \Re\left\{ \mathbf{A}_{2}\right\} \\
\hdashline\Im\left\{ \mathbf{A}_{1}\right\}  & \Im\left\{ \mathbf{A}_{2}\right\} 
\end{array}\right],\label{eq:A}
\end{equation}
where $\mathbf{A}_{1}\in\mathbb{C}^{N\times N}$ and $\mathbf{A}_{2}\in\mathbb{C}^{N\times(N-1)}$.
Because the real and imaginary parts of $\mathbf{A}_{1}$ are multiplied
by the $N$ diagonal entries of $\mathbf{B}$ in \eqref{eq:system3},
 to build the equivalence between \eqref{eq:system2} and \eqref{eq:system3},
$\mathbf{A}_{1}$ should be given by 
\begin{equation}
\mathbf{A}_{1}=\text{diag}\left(\boldsymbol{\alpha}\right).\label{eq:A1}
\end{equation}
On the other hand, the real and imaginary parts of $\mathbf{A}_{2}$
are multiplied by the $N-1$  off-diagonal entries of $\mathbf{B}$
including $\left[\mathbf{B}\right]_{n_{1},m_{1}},\ldots,\left[\mathbf{B}\right]_{n_{N-1},m_{N-1}}$.
For this reason, the structure of $\mathbf{A}_{2}$ depends on how
the RIS ports are connected to each other, that is the specific graph
associated with the tree-connected RIS. Accordingly, the entry of
the $\ell$th column of $\mathbf{A}_{2}$ for $\ell=1,\ldots,N-1$
is given by 
\begin{equation}
\left[\mathbf{A}_{2}\right]{}_{k,\ell}=\begin{cases}
[\boldsymbol{\alpha}]_{n_{\ell}} & \textrm{if}\:k=m_{\ell}\\{}
[\boldsymbol{\alpha}]_{m_{\ell}} & \textrm{if}\:k=n_{\ell}\\
0 & \textrm{otherwise}
\end{cases}.\label{eq:A2}
\end{equation}

Remarkably, the constraint on the tree-connected architecture, i.e., the position of the non-zero elements of $\mathbf{B}$, is dealt with by properly constructing $\mathbf{A}_{2}$ as in \eqref{eq:A2}.

The following proposition provides a useful property of tree-connected
RIS, which we can leverage to solve the linear equations \eqref{eq:system3}.

\begin{proposition} For any tree-connected RIS with $N$  elements,
the coefficient matrix $\mathbf{A}\in\mathbb{R}^{2N\times(2N-1)}$
has full column rank, i.e., $r\left(\mathbf{A}\right)=2N-1$. \label{pro:rank1}
\end{proposition} \begin{proof} Please refer to Appendix~C. \end{proof}

According to Proposition~\ref{pro:rank1}, the linear equations \eqref{eq:system3}
has one solution or no solution \cite{str93}. Furthermore, we show
that it has exactly one solution by using the following proposition.

\begin{proposition}  For any tree-connected RIS with $N$ elements,
the augmented matrix $\left[\mathbf{A}|\mathbf{b}\right]\in\mathbb{R}^{2N\times2N}$
has rank $r\left(\left[\mathbf{A}|\mathbf{b}\right]\right)=2N-1$.
\label{pro:rank2} \end{proposition} \begin{proof} Please refer
to Appendix~D. \end{proof}

\begin{algorithm}[t]
\begin{algorithmic}[1]
\setstretch{1.2}
\REQUIRE $\mathbf{h}_{RI}\in\mathbb{C}^{1\times N}$, $\mathbf{H}_{IT}\in\mathbb{C}^{N\times M}$.
\ENSURE $\mathbf{B}$.
\STATE{Obtain $\hat{\mathbf{h}}_{RI}=\frac{\mathbf{h}_{RI}}{\left\Vert\mathbf{h}_{RI}\right\Vert}$, $\mathbf{u}_{IT}=\mathbf{u}_{\text{max}}\left(\mathbf{H}_{IT}\right)$.}
\STATE{Obtain $\boldsymbol{\alpha}= jZ_{0}\left(\mathbf{u}_{IT}+\hat{\mathbf{h}}_{RI}^H\right)$, $\boldsymbol{\beta}=\mathbf{u}_{IT}-\hat{\mathbf{h}}_{RI}^H$.}
\STATE{Obtain $\mathbf{A}_1$ and $\mathbf{A}_2$ by \eqref{eq:A1} and \eqref{eq:A2}, respectively.}
\STATE{Obtain $\mathbf{A}$ and $\mathbf{b}$ by \eqref{eq:A} and \eqref{eq:b}, respectively.}
\STATE{Obtain $\mathbf{x}=\left(\mathbf{A}^T\mathbf{A}\right)^{-1}\mathbf{A}^T\mathbf{b}$.}
\STATE{Obtain $\mathbf{B}$ by \eqref{eq:x}.}
\end{algorithmic}
\caption{Closed-form global optimization of tree-connected RIS for MISO systems}
\label{alg:B-design}
\end{algorithm}

According to the Rouche-Capelli theorem \cite{str93}, a system of
linear equations is consistent if the rank of the coefficient matrix
$\mathbf{A}$ is equal to the rank of the augmented matrix $\left[\mathbf{A}|\mathbf{b}\right]$.
Thus, according to Proposition~\ref{pro:rank2}, the system of linear
equations \eqref{eq:system3} is consistent and has exactly one solution.
As a result, we can find the solution of \eqref{eq:system3} as
$\mathbf{x}=\mathbf{A}^{\dagger}\mathbf{b}$, where $\mathbf{A}^{\dagger}=\left(\mathbf{A}^{T}\mathbf{A}\right)^{-1}\mathbf{A}^{T}$
is the Moore-Penrose pseudo-inverse of $\mathbf{A}$, and subsequently
the optimal susceptance matrix $\mathbf{B}$ of the tree-connected
RIS can be found through \eqref{eq:x}. 

To conclude, we prove that for any tree-connected RIS there is always
one and only one global optimal solution for the susceptance matrix
$\mathbf{B}$ to achieve the upper bound $\bar{P}_{R}$ for any channel
realization.  In Alg.~\ref{alg:B-design}, we summarize the steps
required to optimize tree-connected RIS for \gls{miso} systems.
Interestingly, Alg.~\ref{alg:B-design} can be used to globally optimize tree-connected RISs in single-user \gls{mimo} and multi-user \gls{miso} systems.
More precisely, in single-user \gls{mimo} systems using single-stream transmission, the received signal power is maximized by applying Alg.~\ref{alg:B-design} to the vectors $\mathbf{v}_{RI}$ and $\mathbf{u}_{IT}$, with $\mathbf{v}_{RI}$ being the dominant right singular vector of $\mathbf{H}_{RI}$.
In addition, in multi-user \gls{miso} systems, the sum power is maximized by applying Alg.~\ref{alg:B-design} to $\mathbf{t}_{RI}$ and $\mathbf{u}_{IT}$, with $\mathbf{t}_{RI}$ being the dominant right singular vector of $\mathbf{G}_{RI}=\left[\mathbf{h}_{RI,1}^H,\ldots,\mathbf{h}_{RI,U}^H\right]^H$, where $\mathbf{h}_{RI,u}\in\mathbb{C}^{1\times N}$ denotes the channel from the RIS to the $u$th receiver.
The computational complexity of Alg.~\ref{alg:B-design} is driven by the complexity of Step~5, i.e., it is the complexity of computing the inverse of $\mathbf{A}^T\mathbf{A}$.
Since $\mathbf{A}^T\mathbf{A}\in\mathbb{R}^{(2N-1)\times(2N-1)}$, the computational complexity of Alg.~\ref{alg:B-design} is $\mathcal{O}(8N^3)$.

\subsection{Forest-Connected RIS Optimization}

The forest-connected RIS optimization to maximize the received signal
power in MISO systems is formulated as

\begin{align}
\underset{\boldsymbol{\Theta},\mathbf{w}}{\mathsf{\mathrm{max}}}\;\; & P_{T}\left\vert \mathbf{h}_{RI}\boldsymbol{\Theta}\mathbf{H}_{IT}\mathbf{w}\right\vert ^{2}\label{eq:received power FOREST O}\\
\mathsf{\mathrm{s.t.}}\;\;\; & \boldsymbol{\Theta}=\left(\mathbf{I}+jZ_{0}\mathbf{B}\right)^{-1}\left(\mathbf{I}-jZ_{0}\mathbf{B}\right),\label{eq:F C1}\\
 & \mathbf{B}=\mathrm{diag}\left(\mathbf{B}_{1},\ldots,\mathbf{B}_{G}\right),\label{eq:F C2}\\
 & \mathbf{B}_{g}=\mathbf{B}_{g}^{T},\:\mathbf{B}_{g}\in\mathcal{B}_{\mathcal{G},g},\:\forall g,\label{eq:F C4}\\
 & \Vert\mathbf{w}\Vert=1,\label{eq:F C5}
\end{align}
where the susceptance matrix $\mathbf{B}$ of the forest-connected
RIS is a block diagonal matrix with $\mathbf{B}_{g}$ being the $g$th
block and $\mathcal{B}_{\mathcal{G},g}$ represents the set of the
$g$th block written as
\begin{equation}
\mathcal{B}_{\mathcal{G},g}=\left\{ \mathbf{B}_{g}|\left[\mathbf{B}_{g}\right]{}_{n,m}=0,n\neq m,n,m\in\mathcal{V}_{g},\left(n,m\right)\notin\mathcal{E}_{g}\right\} ,
\end{equation}
where $\mathcal{V}_{g}$ and $\mathcal{E}_{g}$ are the vertex and
edge sets of the graph for the $g$th group for $g=1,\ldots,G$. Specifically,
$\mathcal{V}_{g}$ and $\mathcal{E}_{g}$ can be found by \eqref{eq:VFOREST-TRI}
and \eqref{eq:EFOREST-TRI} when the tridiagonal architecture is used
for the $g$th group, or by \eqref{eq:VFOREST-ARR} and \eqref{eq:EFOREST-ARR}
when the arrowhead architecture is used.

Different from optimizing the tree-connected RIS, it is hard to find
a  tight upper bound on the received signal power  in forest-connected
RIS-aided \gls{miso} systems due to block diagonal characteristics
of the susceptance matrix, which makes it hard to  derive a closed-form
global optimal solution. To solve the challenging forest-connected
RIS optimization problem,  we propose an iterative approach,
where the precoder $\mathbf{w}$ and the RIS susceptance matrix $\mathbf{B}$
are alternatively optimized.

\subsubsection{Optimizing $\mathbf{B}$ with Fixed $\mathbf{w}$}

When the precoder $\mathbf{w}$ is fixed, we introduce $\mathbf{h}_{IT}^{\mathrm{eff}}=\mathbf{H}_{IT}\mathbf{w}$
as the effective channel between the transmitter and RIS so that the
forest-connected RIS-aided MISO system is equivalently converted to
a forest-connected RIS-aided SISO system. Thus, we can equivalently
rewrite the problem \eqref{eq:received power FOREST O}-\eqref{eq:F C5}
as
\begin{align}
\underset{\boldsymbol{\Theta}_{g}}{\mathsf{\mathrm{max}}}\;\; & P_{T}\left\vert \sum_{g=1}^{G}\mathbf{h}_{RI,g}\boldsymbol{\Theta}_{g}\mathbf{h}_{IT,g}^{\mathrm{eff}}\right\vert ^{2}\label{eq:received power FOREST O-E1}\\
\mathsf{\mathrm{s.t.}}\;\;\; & \boldsymbol{\Theta}_{g}=\left(\mathbf{I}+jZ_{0}\mathbf{B}_{g}\right)^{-1}\left(\mathbf{I}-jZ_{0}\mathbf{B}_{g}\right),\label{eq:F C1-E1}\\
 & \mathbf{B}_{g}=\mathbf{B}_{g}^{T},\:\mathbf{B}_{g}\in\mathcal{B}_{\mathcal{G},g},\:\forall g,\label{eq:F C4-E1}\\
 & \Vert\mathbf{w}\Vert=1,\label{eq:F C5-E1}
\end{align}
where $\mathbf{h}_{RI,g}\in\mathbb{C}^{1\times N_{G}}$ and $\mathbf{h}_{IT,g}^{\mathrm{eff}}\in\mathbb{C}^{N_{G}\times1}$
consist of the $N_{G}$ entries of $\mathbf{h}_{RI}$ and $\mathbf{h}_{IT}^{\mathrm{eff}}$
corresponding to the RIS elements of the $g$th group, respectively
\cite{she20}. To solve the problem \eqref{eq:received power FOREST O-E1}-\eqref{eq:F C5-E1},
we first consider the case that the graph for each group is complete,
that is the group-connected RIS where $\mathbf{B}_{g}$ can be arbitrary
symmetric matrix. As shown in \cite{she20} and \cite{ner22}, the
received signal power in the group-connected RIS-aided SISO system,
i.e. $P_{T}\left\vert \sum_{g=1}^{G}\mathbf{h}_{RI,g}\boldsymbol{\Theta}_{g}\mathbf{h}_{IT,g}^{\mathrm{eff}}\right\vert ^{2}$,
is upper bounded by
\begin{equation}
\bar{P}_{R}^{\mathrm{eff}}=P_{T}\left(\sum_{g=1}^{G}\left\Vert \mathbf{h}_{RI,g}\right\Vert \left\Vert \mathbf{h}_{IT,g}^{\mathrm{eff}}\right\Vert \right)^{2}.\label{eq:PR-UB-group}
\end{equation}
 The key to achieve such upper bound  \eqref{eq:PR-UB-group} is to
find complex symmetric unitary matrices $\boldsymbol{\Theta}_{g}$
for $g=1,\ldots,G$ satisfying
\begin{equation}
\hat{\mathbf{h}}_{RI,g}^{H}=\boldsymbol{\Theta}_{g}\hat{\mathbf{h}}_{IT,g}^{\mathrm{eff}},\:\forall g,\label{eq:optimal group}
\end{equation}
where $\hat{\mathbf{h}}_{RI,g}=\mathbf{h}_{RI,g}/\Vert\mathbf{h}_{RI,g}\Vert$
and $\hat{\mathbf{h}}_{IT,g}^{\mathrm{eff}}=\mathbf{h}_{IT,g}^{\mathrm{eff}}/\Vert\mathbf{h}_{IT,g}^{\mathrm{eff}}\Vert$.
Recall that each group in the forest-connected RIS utilizes the tree-connected
architecture. Thus, referring to the tree-connected RIS optimization
shown in Section IV.A, it is straightforward to show that for each
group of the forest-connected RIS there is always one and only one
solution for $\mathbf{B}_{g}$ whose corresponding $\boldsymbol{\Theta}_{g}$
satisfying \eqref{eq:optimal group}, so that the upper bound \eqref{eq:PR-UB-group}
can be achieved. Moreover,  the optimal $\mathbf{B}_{g}$ for each
group can be found by Alg.~\ref{alg:B-design} with the input of
the channels $\mathbf{h}_{RI,g}$ and $\mathbf{h}_{IT,g}^{\mathrm{eff}}$.
Therefore,  the susceptance matrix $\mathbf{B}$ of the forest-connected
RIS can be globally optimized to  achieve the upper bound \eqref{eq:PR-UB-group}
when $\mathbf{w}$ is fixed.

\begin{algorithm}[t]
\begin{algorithmic}[1]
\setstretch{1.2}
\REQUIRE $\mathbf{h}_{RI}\in\mathbb{C}^{1\times N}$, $\mathbf{H}_{IT}\in\mathbb{C}^{N\times M}$.
\ENSURE $\mathbf{B}$ and $\mathbf{w}$.
\STATE{Initialize $\mathbf{w}$.}
\WHILE{no convergence of objective \eqref{eq:received power FOREST O}}
\STATE{Update $\mathbf{B}_g$ by Alg. 1 with input $\mathbf{h}_{RI,g}$, $\mathbf{h}_{IT,g}^{\mathrm{eff}}$, $\forall g$.}
\STATE{Obtain $\boldsymbol{\Theta}$ from $\mathbf{B}=\mathrm{diag}\left(\mathbf{B}_{1},\ldots,\mathbf{B}_{G}\right)$ by \eqref{eq:T(B)}.}
\STATE{Update $\mathbf{w}=\left(\mathbf{h}_{RI}\boldsymbol{\Theta}\mathbf{H}_{IT}\right)^H/\Vert\mathbf{h}_{RI}\boldsymbol{\Theta}\mathbf{H}_{IT}\Vert$.}
\ENDWHILE
\end{algorithmic}
\caption{Optimization of forest-connected RIS for MISO systems}
\label{alg:B-design-forest}
\end{algorithm}

\subsubsection{Optimizing $\mathbf{w}$ with Fixed $\mathbf{B}$}

When the RIS susceptance matrix  $\mathbf{B}$ is fixed, it is obvious
that the optimal $\mathbf{w}$ is the MRT given by  $\mathbf{w}=(\mathbf{h}_{RI}\boldsymbol{\Theta}\mathbf{H}_{IT})^{H}/\left\Vert \mathbf{h}_{RI}\boldsymbol{\Theta}\mathbf{H}_{IT}\right\Vert $,
where $\boldsymbol{\Theta}$ can be computed by \eqref{eq:T(B)} as
a function of $\mathbf{B}$. 

The precoder $\mathbf{w}$ and the RIS susceptance matrix $\mathbf{B}$
are alternatively optimized until convergence of the objective \eqref{eq:received power FOREST O}.
We summarize the steps to optimize forest-connected RIS in MISO systems
in Alg.~\ref{alg:B-design-forest}. The iterative algorithm is initialized
by randomly setting $\mathbf{w}$ to a feasible value. In addition,
the convergence of Alg.~\ref{alg:B-design-forest} is guaranteed
by the following two facts. First, at each iteration, the objective
given by the received signal power $P_{R}$ is non-decreasing. Second,
the objective function is upper bounded by $P_{T}\left\Vert \mathbf{h}_{RI}\right\Vert ^{2}\left\Vert \mathbf{H}_{IT}\right\Vert ^{2}$
because of the sub-multiplicativity of the spectral norm.
The computational complexity per iteration of Alg.~\ref{alg:B-design-forest} is driven by the complexity of Step~3, i.e., it is the complexity of applying $N/N_G$ times Alg.~\ref{alg:B-design} to an $N_G$-element BD-RIS.
Thus, the computational complexity per iteration of Alg.~\ref{alg:B-design-forest} is $\mathcal{O}(8NN_G^2)$.

\section{Performance Evaluation}
\label{sec:results}

In this section, we provide the performance evaluation for the proposed
tree- and forest-connected RIS. We consider a two-dimensional coordinate
system, as shown in Fig. \ref{fig:system}.
The transmitter and receiver
are located at $(0,0)$ and $(52,0)$ in meters (m), respectively.
The RIS is located at $(50,2)$ and it is equipped with $N$ elements.
For the large-scale path loss, we use the distance-dependent path
loss  modeled as $L_{ij}(d_{ij})=L_{0}(d_{ij}/D_{0})^{-\alpha_{ij}}$,
where $L_{0}$ is the reference path loss at distance $D_{0}=1$ m,
$d_{ij}$ is the distance, and $\alpha_{ij}$ is the path loss exponent
for $ij\in\{RI,IT\}$. We set $L_{0}=-30$ dB, $\alpha_{RI}=2.8$,
$\alpha_{IT}=2$, and $P_{T}=10$ mW. For the small-scale fading,
we assume that the channel from the RIS to the receiver is Rayleigh
fading. The channel from the transmitter to the RIS is modeled with
 Rician fading, given by 
\begin{equation}
\mathbf{H}_{IT}=\sqrt{L_{IT}}\left(\sqrt{\frac{K}{1+K}}\mathbf{H}_{IT}^{\mathrm{LoS}}+\sqrt{\frac{1}{1+K}}\mathbf{H}_{IT}^{\mathrm{NLoS}}\right),\label{eq:Rician Channel Model}
\end{equation}
where $K$ refers to the Rician factor, while $\mathbf{H}_{IT}^{\mathrm{LoS}}$
and $\mathrm{vec}(\mathbf{H}_{IT}^{\mathrm{NLoS}})\sim\mathcal{CN}\left(\boldsymbol{0},\mathbf{I}\right)$
represent the small-scale \gls{los} and \gls{nlos} (Rayleigh
fading) components, respectively. We consider two scenarios, with
Rician factor $K=0$ dB and $K=10$ dB.

\begin{figure}[t]
\centering{}
\includegraphics[width=0.39\textwidth]{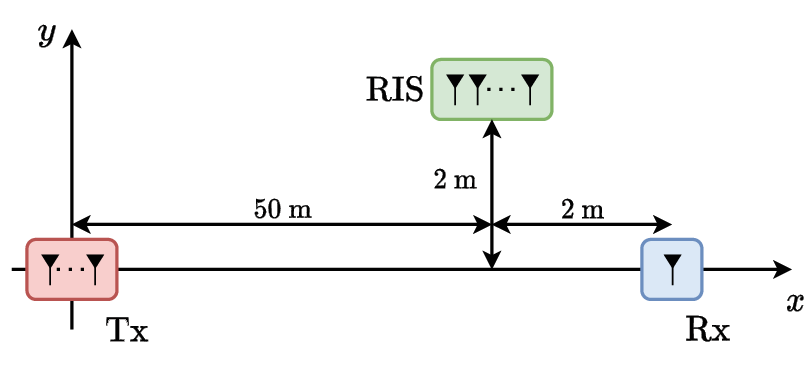}
\caption{Two-dimensional coordinate system for the RIS-aided MISO system.}
\label{fig:system} 
\end{figure}

\begin{figure*}[t]
\centering{}
\includegraphics[width=0.39\textwidth]{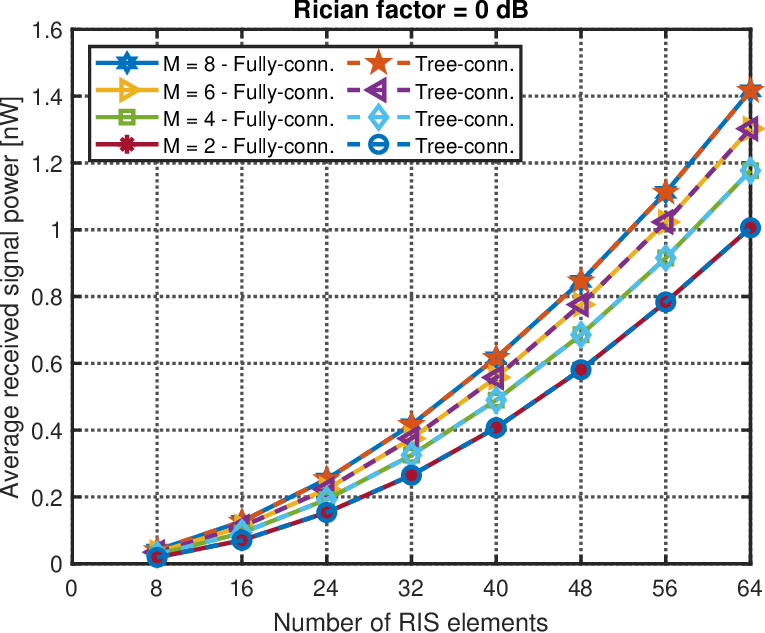}
\includegraphics[width=0.39\textwidth]{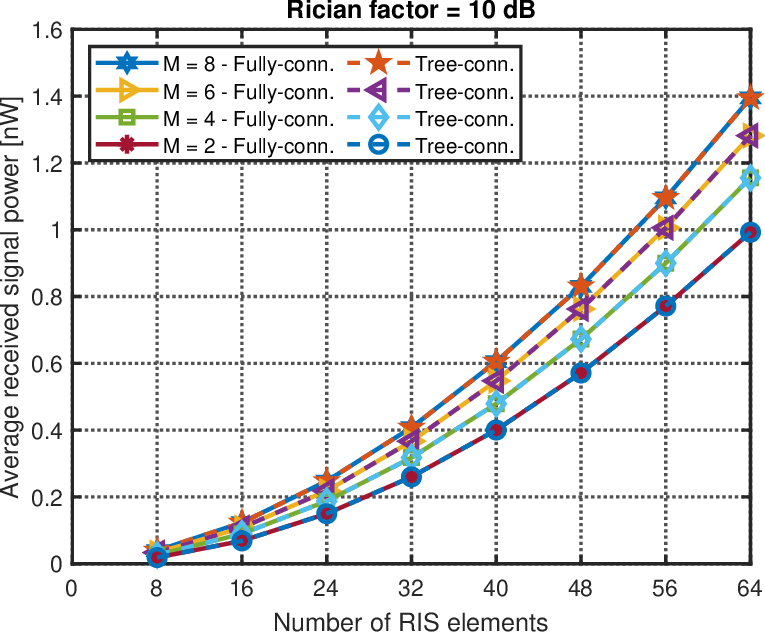}
\caption{Received signal power in MISO systems aided by fully- and tree-connected
RIS.}
\label{fig:perf-miso-tree}
\end{figure*}

\begin{figure*}[t]
\centering{}
\includegraphics[width=0.39\textwidth]{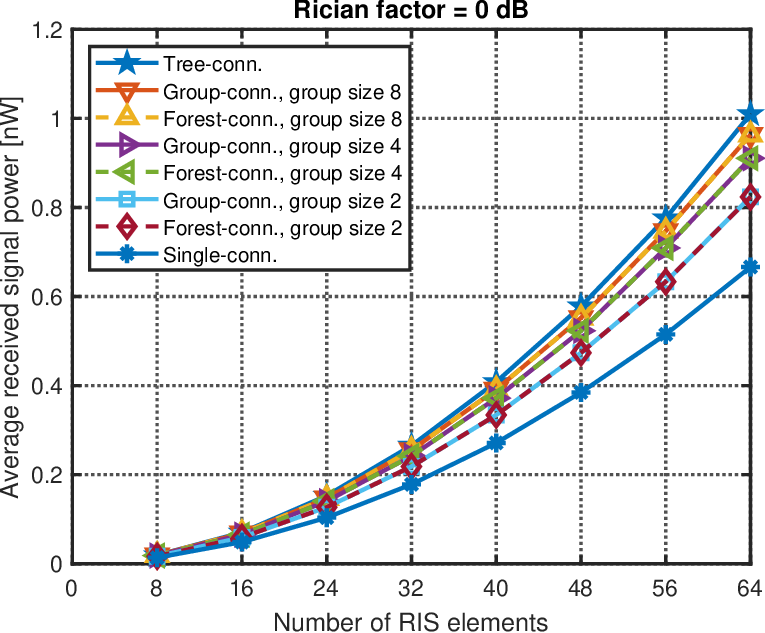}
\includegraphics[width=0.39\textwidth]{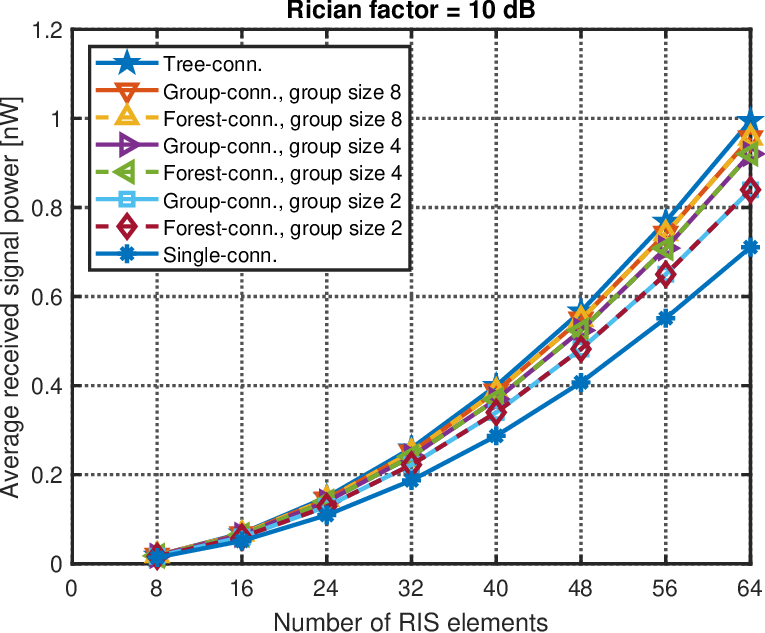}
\caption{Received signal power in MISO systems aided by tree-, group-, forest-,
and single-connected RIS, with $M=2$.}
\label{fig:perf-miso-forest-M2}
\end{figure*}

\subsection{Tree-Connected RIS-Aided MISO Systems}

We first evaluate the performance of the tree-connected RIS to verify
that it is MISO optimal. Specifically, the tree-connected RIS is 
optimized through Alg.~\ref{alg:B-design} and compared with the
performance upper bound \eqref{eq:PR-UB} achieved by the fully-connected
RIS. We recall that in the fully-connected RIS, all the RIS ports are
connected with each other, so that $\mathbf{B}$ is an arbitrary symmetric
matrix and $\boldsymbol{\Theta}$ satisfies that
\begin{equation}
\boldsymbol{\Theta}^{H}\boldsymbol{\Theta}=\boldsymbol{\mathrm{I}},\boldsymbol{\Theta}=\boldsymbol{\Theta}^{T},
\end{equation}
 as shown in \cite{she20}. In Fig.~\ref{fig:perf-miso-tree}, we
provide the received signal power achieved in MISO systems aided by
fully- and tree-connected RIS. We can make the following observations.
\emph{First}, as expected, the tree-connected RIS can always achieve
the performance upper bound. However, it has much lower circuit complexity
than the fully-connected RIS, which will be quantitatively shown in
Section \ref{subsec:complexity}. \emph{Second},  higher received
signal power can be obtained by increasing the number of transmit
antennas $M$ for both fully- and tree-connected RIS. \emph{Last},
Rician fading channels with a lower Rician factor offer richer scattering,
allowing to reach a slightly higher performance.

\begin{figure*}[t]
\centering{}
\includegraphics[width=0.39\textwidth]{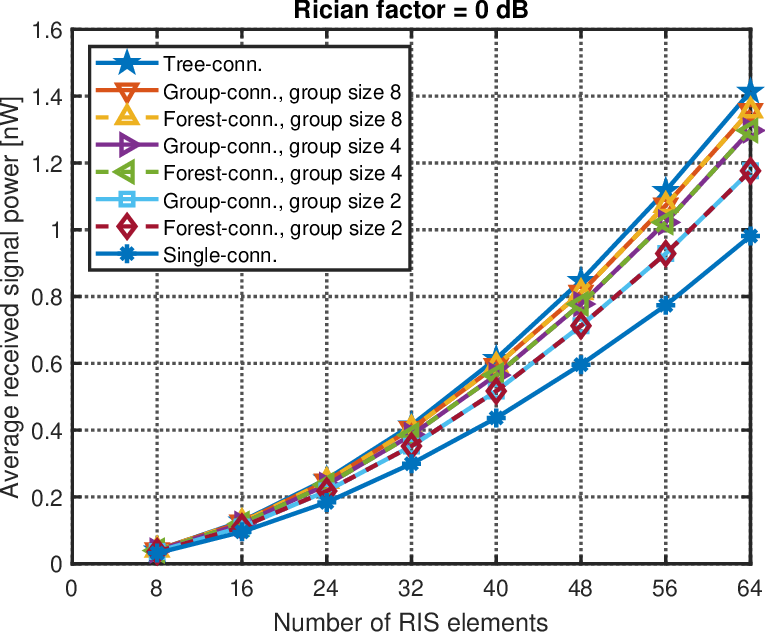}
\includegraphics[width=0.39\textwidth]{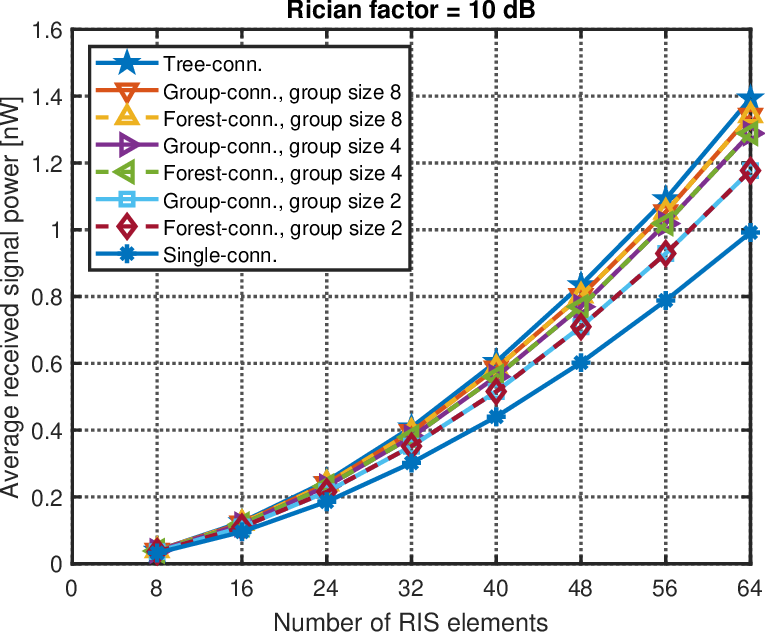}
\caption{Received signal power in MISO systems aided by tree-, group-, forest-,
and single-connected RIS, with $M=8$.}
\label{fig:perf-miso-forest-M8}
\end{figure*}

\subsection{Forest-Connected RIS-Aided MISO Systems}

We next evaluate the performance of the forest-connected RIS. Specifically,
the forest-connected RIS is optimized through Alg.~\ref{alg:B-design-forest}
and it is compared with tree-, group-, and single-connected RIS.
For the group-connected RIS, we recall that  the $N$ elements are
divided into $G$ groups and all the RIS ports within the same group
are connected with each other, so that $\mathbf{B}$ is a symmetric block diagonal matrix and $\boldsymbol{\Theta}$ satisfies
\begin{equation}
\boldsymbol{\Theta}=\mathrm{diag}\left(\boldsymbol{\Theta}_{1},\boldsymbol{\Theta}_{2},\ldots,\boldsymbol{\Theta}_{G}\right),
\end{equation}
\begin{equation}
\boldsymbol{\Theta}_{g}=\boldsymbol{\Theta}_{g}^{T},\:\boldsymbol{\Theta}_{g}^{H}\boldsymbol{\Theta}_{g}=\boldsymbol{\mathrm{I}},\:\forall g,
\end{equation}
as shown in \cite{she20}. In Figs.~\ref{fig:perf-miso-forest-M2}
and \ref{fig:perf-miso-forest-M8}, we provide the received signal
power in MISO systems aided by tree-, group-, forest-, and single-connected
RIS, with $M=2$ and $M=8$, respectively. Group-connected and single-connected
RIS are optimized as proposed in \cite[Alg. 2]{ner22}. We can make
the following observations. 

\emph{First}, the forest-connected RIS can always achieve the same
received signal power as the group-connected RIS with the same group
size. However, it has a much lower circuit complexity than the group-connected
RIS with the same group size, as quantitatively shown in
Section \ref{subsec:complexity}.

\emph{Second}, the forest-connected RIS achieves a higher received
signal power than the single-connected RIS. For example, forest-connected
RIS with group size $8$ brings an improvement in the received signal
power of $44.6\%$, when $M=2$, $N=64$, and $K=0$ dB. On the other
hand, it achieves a lower received signal power than the tree-connected
RIS due to the simplified circuit complexity. Therefore, it is shown
that the forest-connected RIS achieves a good performance-complexity
trade-off between the single- and tree-connected RIS.

\emph{Third}, the received signal power increases with the group size
in the forest-connected RIS, which is because increasing the group
size can provide more flexibility to the BD-RIS. 

\textit{Fourth}, higher received signal power can be obtained by increasing
the number of transmit antennas $M$.

\emph{Last}, BD-RIS is particularly beneficial over single-connected
RIS in the presence of fading channels with lower Rician factors,
in agreement with \cite{she20}.

In Fig.~\ref{fig:rate}, we also provide the achievable rate in MISO systems aided by tree-, group-, forest-, and single-connected RIS, with $M=2$, $K=0$~dB, and \gls{awgn} power $\sigma^2=-80$~dBm.
With $N=64$, the tree-connected RIS improves the achievable rate by $10\%$ over the single-connected RIS.
In addition, the performance of a 64-element single-connected RIS can be obtained by a 52-element tree-connected RIS, which is beneficial for reducing by $19\%$ the area of RIS.

\subsection{Circuit Topology Complexity}

\label{subsec:complexity}

We finally evaluate the circuit topology complexity of the proposed
tree- and forest-connected RIS. As analyzed in Section III, the circuit
topology complexity in terms of the number of the tunable admittance
components of the tree- and forest-connected RIS is $2N-1$ and $N(2-1/N_{G})$,
respectively. For comparison, we also consider the circuit topology
complexity of fully-, group-, and single-connected RIS, given by
$N(N+1)/2$, $N(N_{G}+1)/2$, and $N$, respectively \cite{she20}.
In Fig.~\ref{fig:complexity}, we provide the number of tunable admittance
components in the fully-, group, tree-, forest-, and single-connected
RIS. We can make the following observations.

\textit{First}, compared with the fully-connected RIS, the tree-connected
RIS has much lower circuit topology complexity. For example, the number
of tunable admittance components is decreased by $16.4$ times when
$N=64$ in the tree-connected RIS compared to the fully-connected
RIS, but they achieve the same performance.

\textit{Second}, compared with the group-connected RIS with same group
size, the forest-connected RIS has much lower circuit topology complexity.
For example, the number of tunable admittance components is reduced
by $2.4$ times when $N=64$ in the forest-connected RIS compared to the group-connected RIS with group
size $8$, but they achieve the same performance.

\textit{Third}, increasing the group size in the forest-connected
RIS will increase the number of tunable admittance components while
also enhancing the performance. 

\textit{Fourth}, compared with the single-connected RIS, the tree-
and forest-connected RIS introduce appropriate complexity but achieves
an improvement in the received power.

To conclude, we demonstrate that the benefit of the tree-connected (resp.
forest-connected) RIS over the fully-connected (resp. group-connected)
RIS lies in their highly simplified circuit complexity while maintaining
optimal performance. Therefore, the proposed tree- and forest-connected
architectures significantly improve the performance-complexity trade-off
over existing BD-RIS architectures. 

\begin{figure}[t]
\centering{}
\includegraphics[width=0.39\textwidth]{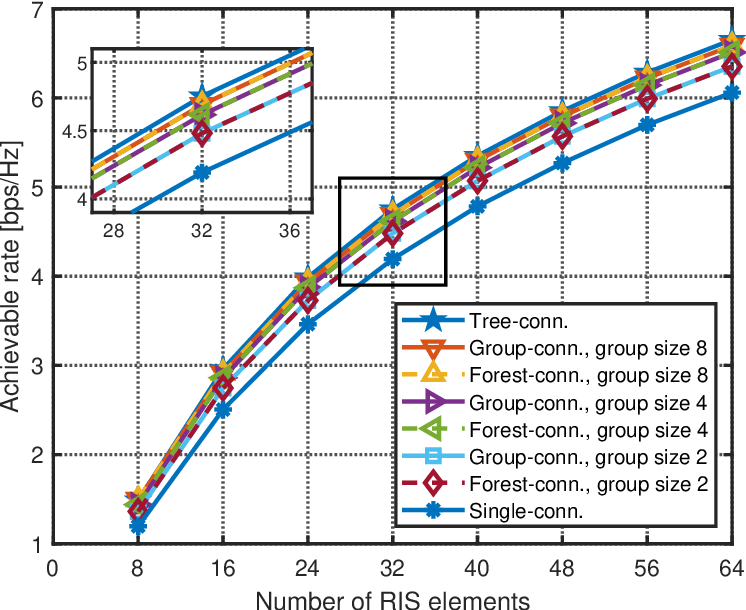}
\caption{Achievable rate in MISO systems aided by tree-, group-, forest-,
and single-connected RIS, with $M=2$ and $K=0$ dB.}
\label{fig:rate} 
\end{figure}

\begin{figure}[t]
\centering{}
\includegraphics[width=0.39\textwidth]{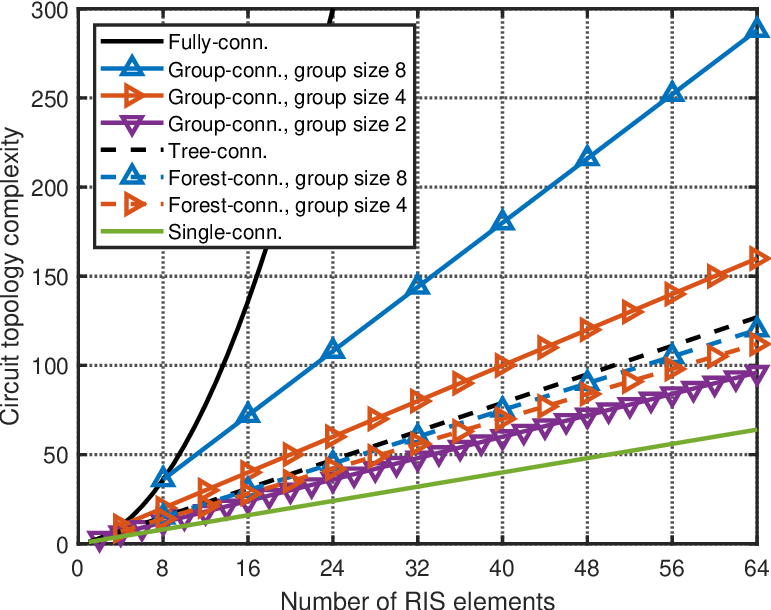}
\caption{Circuit topology complexity, i.e., the number of tunable admittance
components, of fully-, group-, tree-, forest-, and single-connected
RIS.}
\label{fig:complexity} 
\end{figure}

\section{Conclusion}
\label{sec:conclusion}

We propose novel modeling, architecture design, and optimization for
BD-RIS by utilizing graph theory. In particular, we model BD-RIS architectures
as graphs, capturing the presence of interconnections between the
RIS elements. Through this modeling, we prove that a BD-RIS achieves
the performance upper bound in MISO systems if and only if its associated
graph is connected. This remarkable result allows us to characterize
the least complex BD-RIS architectures able to achieve the performance
upper bound in MISO systems, denoted as tree-connected RIS. We also
propose forest-connected RIS to bridge between the single-connected
and the tree-connected architectures. To optimize these novel BD-RIS architectures,
we derive a closed-form global optimal solution for tree-connected
RIS, and an iterative algorithm for forest-connected RIS. Numerical
results confirm that tree-connected (resp. forest-connected) RIS
achieve the same performance as fully-connected (resp. group-connected)
RIS, with a significantly reduced circuit complexity by up to 16.4 times.
We leave the optimization and performance evaluation of the proposed tree- and forest-connected RISs in other scenarios such as multi-user systems as the object of
future research.

The proposed graph theoretical modeling of BD-RIS is expected to
promote significant advancements in the growth of BD-RIS.
This modeling can be used to explore the vast design space of possible architectures
and evaluate them in terms of achievable performance and circuit complexity.
This will enable the systematic development of new BD-RIS architectures
with a more favorable balance between performance and complexity, and the implementation of BD-RIS prototypes.

\section*{Appendix}

\subsection{Proof of Lemma~\ref{lem:connected}}

We prove the necessary condition by showing that if $\mathcal{G}$
is not a connected graph, the corresponding BD-RIS is not MISO optimal.
If $\mathcal{G}$ is disconnected, it has $C\geq2$ connected components,
where the $c$th component includes $N^{(c)}$ ports, for $c=1,\ldots,C$,
with $\sum_{c=1}^{C}N^{(c)}=N$. With no loss of generality, we assume
that each component includes adjacent ports. Thus, the admittance
matrix $\mathbf{Y}$ given by \eqref{eq:Yij} is block diagonal, with
the $c$th block having dimensions $N^{(c)}\times N^{(c)}$, for $c=1,\ldots,C$.
As a consequence of \eqref{eq:T(Y)}, also the scattering matrix $\boldsymbol{\Theta}$
is block diagonal, and writes as $\boldsymbol{\Theta}=\mathrm{diag}\left(\boldsymbol{\Theta}_{1},\ldots,\boldsymbol{\Theta}_{C}\right)$,
where $\boldsymbol{\Theta}_{c}\in\mathbb{C}^{N^{(c)}\times N^{(c)}}$,
for $c=1,\ldots,C$.  Thus, the received signal power \eqref{eq:received power}
in the case of a disconnected graph $\mathcal{G}$ is 
\begin{equation}
P_{R}^{\textrm{\ensuremath{\mathrm{DisC}}}}=P_{T}\left\vert \sum_{c=1}^{C}\mathbf{h}_{RI,c}\boldsymbol{\Theta}_{c}\mathbf{H}_{IT,c}\mathbf{w}\right\vert ^{2},
\end{equation}
where $\mathbf{h}_{RI,c}\in\mathbb{C}^{1\times N^{(c)}}$ and $\mathbf{H}_{IT,c}\in\mathbb{C}^{N^{(c)}\times M}$
contain the $N^{(c)}$ elements of $\mathbf{h}_{RI}$ and rows of
$\mathbf{H}_{IT}$ corresponding to the $N^{(c)}$ RIS elements grouped
into the $c$th component, respectively. The received signal power
$P_{R}^{\textrm{\ensuremath{\mathrm{DisC}}}}$ is upper bounded by
\begin{equation}
\bar{P}_{R}^{\textrm{DisC}}=P_{T}\left(\sum_{c=1}^{C}\left\Vert \mathbf{h}_{RI,c}\right\Vert \left\Vert \mathbf{H}_{IT,c}\mathbf{w}\right\Vert \right)^{2},
\end{equation}
following the Cauchy-Schwarz inequality and that $\boldsymbol{\Theta}_{c}^{H}\boldsymbol{\Theta}_{c}=\boldsymbol{\mathrm{I}}$,
for $c=1,\ldots,C$. Furthermore, it holds that 
\begin{equation}
\bar{P}_{R}^{\textrm{DisC}}\overset{(a)}{\leq}P_{T}\left\Vert \mathbf{h}_{RI}\right\Vert ^{2}\left\Vert \mathbf{H}_{IT}\mathbf{w}\right\Vert ^{2}\overset{(b)}{\leq}\bar{P}_{R},
\end{equation}
where $(a)$ follows by the Cauchy-Schwarz inequality and $(b)$ follows
by the definition of the spectral norm. Note that $(b)$ is an equality
if and only if $\mathbf{w}$ is the dominant right singular vector
of $\mathbf{H}_{IT}$. However,  $(a)$ is in general a strict inequality
if $\mathbf{w}$ is the dominant right singular vector of $\mathbf{H}_{IT}$,
which thus proves the necessary condition of the lemma. 

On the other hand, we prove the sufficient condition by showing that
if $\mathcal{G}$ is a connected graph, the corresponding BD-RIS is
MISO optimal. This is straightforward to prove because 1) for a connected
graph, we can always remove some edges to make it a tree, i.e. setting
the corresponding off-diagonal entries of $\mathbf{B}$ zero, 2) there
exist one and only one solution for any RIS whose associated graph
is a tree to achieve the performance upper bound (i.e., to be MISO
optimal) as shown in Section IV. A.

\subsection{Proof of Proposition~\ref{pro:N-1}}

Applying Lemma~\ref{lem:connected}, we need to prove that a connected
graph with $N$ vertices has at least $N-1$ edges. This is achieved
by induction. The base case is easily verified: a connected graph
with a single vertex has at least zero edges. As the induction step,
we consider a connected graph with $N$ vertices. From this graph,
we remove edges (at least one) until we obtain a disconnected graph
with two connected components. Assuming the two components have $K$
and $N-K$ vertices, they have at least $K-1$ and $N-K-1$ edges
by the induction hypothesis, respectively. Since we removed at least
one edge to disconnect the graph, it had originally at least $(K-1)+(N-K-1)+1=N-1$
edges.

\subsection{Proof of Proposition~\ref{pro:rank1}}

This proof is conducted by induction. As the base case, we consider
the only tree-connected RIS that can be realized with $N=2$ ports.
This RIS includes two tunable admittance components connecting the
two ports to ground and a further tunable admittance connecting the
two ports to each other.

Based on the susceptance matrix $\mathbf{B}\in\mathbb{R}^{2\times2}$,
the left-hand side of the system \eqref{eq:system3} is built by setting
$\mathbf{x}=[[\mathbf{B}]_{1,1},[\mathbf{B}]_{2,2},[\mathbf{B}]_{1,2}]^{T}$.
Accordingly, $\mathbf{A}\in\mathbb{R}^{4\times3}$ is given by \eqref{eq:A}
with $\mathbf{A}_{1}\in\mathbb{C}^{2\times2}$ as in \eqref{eq:A1}
and $\mathbf{A}_{2}\in\mathbb{C}^{2\times1}$ given by $\mathbf{A}_{2}=[[\boldsymbol{\alpha}]_{2},[\boldsymbol{\alpha}]_{1}]^{T}$.
Thus, it is easy to recognize that $\mathbf{A}$ has in general full
column rank, i.e., $r\left(\mathbf{A}\right)=3$. The proposition
is hence verified for the case $N=2$.

As the induction step, we prove that if the proposition is valid for
RISs with $N-1$ ports, it also holds for RISs with $N$ ports. We
consider a tree-connected RIS with $N-1$ elements, whose coefficient
matrix is given by 
\begin{equation}
\mathbf{A}^{(N-1)}=\left[\begin{array}{cc}
\Re\left\{ \mathbf{A}_{1}^{(N-1)}\right\}  & \Re\left\{ \mathbf{A}_{2}^{(N-1)}\right\} \\
\hdashline\Im\left\{ \mathbf{A}_{1}^{(N-1)}\right\}  & \Im\left\{ \mathbf{A}_{2}^{(N-1)}\right\} 
\end{array}\right].\label{eq:A(N-1)}
\end{equation}
To this $(N-1)$-port tree-connected RIS, we connect an additional
port creating an $N$-port tree-connected RIS. With no loss of generality,
we assume that the additional port is connected to the $(N-1)$th
port through a tunable admittance. In this way, the resulting $N$-port
BD-RIS has the coefficient matrix 
\begin{equation}
\mathbf{A}^{(N)}=\left[\begin{array}{cc}
\Re\left\{ \mathbf{A}_{1}^{(N)}\right\}  & \Re\left\{ \mathbf{A}_{2}^{(N)}\right\} \\
\hdashline\Im\left\{ \mathbf{A}_{1}^{(N)}\right\}  & \Im\left\{ \mathbf{A}_{2}^{(N)}\right\} 
\end{array}\right],\label{eq:A(N)}
\end{equation}
where $\mathbf{A}_{1}^{(N)}\in\mathbb{C}^{N\times N}$ and $\mathbf{A}_{2}^{(N)}\in\mathbb{C}^{N\times(N-1)}$
are given by 
\begin{equation}
\mathbf{A}_{1}^{(N)}=\left[\begin{array}{cc}
\mathbf{A}_{1}^{(N-1)} & \mathbf{0}_{(N-1)\times 1}\\
\hdashline
\mathbf{0}_{1\times (N-1)} & \left[\boldsymbol{\alpha}\right]_{N}
\end{array}\right],\label{eq:A1(N)}
\end{equation}
\begin{equation}
\mathbf{A}_{2}^{(N)}=\left[\begin{array}{cc}
\mathbf{A}_{2}^{(N-1)} & \begin{matrix}\mathbf{0}_{(N-2)\times 1}\\
\left[\boldsymbol{\alpha}\right]_{N}
\end{matrix}\\
\hdashline
\mathbf{0}_{1\times (N-2)} & \left[\boldsymbol{\alpha}\right]_{N-1}
\end{array}\right],\label{eq:A2(N)}
\end{equation}
with $\mathbf{0}_{R\times C}$ denoting an $R\times C$ all-zero matrix.
To prove the induction step, we need to show that $r(\mathbf{A}^{(N-1)})=2(N-1)-1$
implies $r(\mathbf{A}^{(N)})=2N-1$.
To this end, $\mathbf{A}^{(N)}$ is rewritten as 
\begin{equation}
\mathbf{A}^{(N)}\sim\left[\begin{array}{ccc}
\mathbf{A}^{(N-1)} & \mathbf{0}_{(2N-2)\times 1} & \begin{matrix}\mathbf{0}_{(N-2)\times 1}\\
\Re\left\{ \left[\boldsymbol{\alpha}\right]_{N}\right\}\\
\mathbf{0}_{(N-2)\times 1}\\
\Im\left\{ \left[\boldsymbol{\alpha}\right]_{N}\right\} 
\end{matrix}\\
\hdashline
\mathbf{0}_{2\times(2N-3)} & \begin{matrix}\Re\left\{ \left[\boldsymbol{\alpha}\right]_{N}\right\} \\
\Im\left\{ \left[\boldsymbol{\alpha}\right]_{N}\right\} 
\end{matrix} & \begin{matrix}\Re\left\{ \left[\boldsymbol{\alpha}\right]_{N-1}\right\} \\
\Im\left\{ \left[\boldsymbol{\alpha}\right]_{N-1}\right\} 
\end{matrix}
\end{array}\right],\label{eq:A(N)2}
\end{equation}
by applying appropriate row and column swapping operations. From \eqref{eq:A(N)2},
we notice that if $\mathbf{A}^{(N-1)}$ has full column rank and the
vectors $\left[\Re\left\{ \left[\boldsymbol{\alpha}\right]_{N}\right\} ,\Im\left\{ \left[\boldsymbol{\alpha}\right]_{N}\right\} \right]$
and $\left[\Re\left\{ \left[\boldsymbol{\alpha}\right]_{N-1}\right\} ,\Im\left\{ \left[\boldsymbol{\alpha}\right]_{N-1}\right\} \right]$
are linearly independent, then $\mathbf{A}^{(N)}$ has full column
rank. Since $\left[\Re\left\{ \left[\boldsymbol{\alpha}\right]_{N}\right\} ,\Im\left\{ \left[\boldsymbol{\alpha}\right]_{N}\right\} \right]$
and $\left[\Re\left\{ \left[\boldsymbol{\alpha}\right]_{N-1}\right\} ,\Im\left\{ \left[\boldsymbol{\alpha}\right]_{N-1}\right\} \right]$
are in practice linearly independent with probability 1, the induction
step is proven.

\subsection{Proof of Proposition~\ref{pro:rank2}}

To prove that $r\left(\left[\mathbf{A}|\mathbf{b}\right]\right)=2N-1$,
we show that it is always possible to obtain a zero row in $\left[\mathbf{A}|\mathbf{b}\right]$
by applying appropriate row operations.
Firstly, we execute on $\left[\mathbf{A}|\mathbf{b}\right]$ the $N$
row operations 
\begin{equation}
\mathbf{r}_{N+n}=\mathbf{r}_{N+n}-\frac{\Im\left\{ \left[\boldsymbol{\alpha}\right]_{n}\right\} }{\Re\left\{ \left[\boldsymbol{\alpha}\right]_{n}\right\} }\mathbf{r}_{n},
\end{equation}
for $n=1,\ldots,N$, where $\mathbf{r}_{n}$ denotes the $n$th row
of $\left[\mathbf{A}|\mathbf{b}\right]$. The resulting matrix is
given by 
\begin{equation}
\left[\mathbf{A}|\mathbf{b}\right]\sim\left[\begin{array}{ccc}
\Re\left\{ \mathbf{A}_{1}\right\}  & \Re\left\{ \mathbf{A}_{2}\right\}  & \Re\left\{ \boldsymbol{\beta}\right\} \\
\hdashline\mathbf{0}_{N\times N} & \mathbf{A}' & \mathbf{b}'
\end{array}\right],\label{eq:Gaussian1}
\end{equation}
where we introduced $\mathbf{A}'\in\mathbb{R}^{N\times (N-1)}$ and
$\mathbf{b}'\in\mathbb{R}^{N\times1}$. The matrix $\mathbf{A}'$
has exactly two non-zero elements in each column, in the same positions
as the non-zero elements in $\mathbf{A}_{2}$. Specifically,
\begin{equation}
\begin{cases}
\left[\mathbf{A}'\right]_{n_{\ell},\ell} & =\Im\left\{ \left[\boldsymbol{\alpha}\right]_{m_{\ell}}\right\} -\frac{\Im\left\{ \left[\boldsymbol{\alpha}\right]_{n_{\ell}}\right\} }{\Re\left\{ \left[\boldsymbol{\alpha}\right]_{n_{\ell}}\right\} }\Re\left\{ \left[\boldsymbol{\alpha}\right]_{m_{\ell}}\right\} ,\\
\left[\mathbf{A}'\right]_{m_{\ell},\ell} & =\Im\left\{ \left[\boldsymbol{\alpha}\right]_{n_{\ell}}\right\} -\frac{\Im\left\{ \left[\boldsymbol{\alpha}\right]_{m_{\ell}}\right\} }{\Re\left\{ \left[\boldsymbol{\alpha}\right]_{m_{\ell}}\right\} }\Re\left\{ \left[\boldsymbol{\alpha}\right]_{n_{\ell}}\right\} ,\\
\left[\mathbf{A}'\right]_{p,\ell} & =0,\forall p\neq m_{\ell},n_{\ell},
\end{cases}
\end{equation}
for $\ell=1,\ldots,N-1$, with $n_{\ell}$ and $m_{\ell}$ being the
row indexes of the two non-zero elements in the $\ell$th column of
$\mathbf{A}_{2}$. 
The vector $\mathbf{b}'$ writes as 
\begin{equation}
\left[\mathbf{b}'\right]_{n}=\Im\left\{ \left[\boldsymbol{\beta}\right]_{n}\right\} -\frac{\Im\left\{ \left[\boldsymbol{\alpha}\right]_{n}\right\} }{\Re\left\{ \left[\boldsymbol{\alpha}\right]_{n}\right\} }\Re\left\{ \left[\boldsymbol{\beta}\right]_{n}\right\} ,
\end{equation}
for $n=1,\ldots,N$. Secondly, we execute on the resulting $\left[\mathbf{A}|\mathbf{b}\right]$
the $N$ row operations 
\begin{equation}
\mathbf{r}_{N+n}=\Re\left\{ \left[\boldsymbol{\alpha}\right]_{n}\right\} \mathbf{r}_{N+n},
\end{equation}
for $n=1,\ldots,N$, followed by the $N-1$ row operations 
\begin{equation}
\mathbf{r}_{2N}=\mathbf{r}_{2N}+\mathbf{r}_{N+n},
\end{equation}
for $n=1,\ldots,N-1$. The resulting matrix is given by 
\begin{equation}
\left[\mathbf{A}|\mathbf{b}\right]\sim\left[\begin{array}{ccc}
\Re\left\{ \mathbf{A}_{1}\right\}  & \Re\left\{ \mathbf{A}_{2}\right\}  & \Re\left\{ \boldsymbol{\beta}\right\} \\
\hdashline\mathbf{0}_{N\times N} & \mathbf{A}'' & \mathbf{b}''
\end{array}\right],\label{eq:Gaussian2}
\end{equation}
where we introduced $\mathbf{A}''\in\mathbb{R}^{N\times N-1}$ and
$\mathbf{b}''\in\mathbb{R}^{N\times1}$.

We now show that \eqref{eq:Gaussian2}
has rank $2N-1$ since its $2N$th row is all-zero.
To this end, we need to prove that $\left[\mathbf{A}''\right]_{N,\ell}=0$, for $\ell=1,\ldots,N-1$, and that $\left[\mathbf{b}''\right]_{N}=0$.
Indeed, we have
\begin{equation}
\left[\mathbf{A}''\right]_{N,\ell}=\Re\left\{ \left[\boldsymbol{\alpha}\right]_{n_{\ell}}\right\} \left[\mathbf{A}'\right]_{n_{\ell},\ell}+\Re\left\{ \left[\boldsymbol{\alpha}\right]_{m_{\ell}}\right\} \left[\mathbf{A}'\right]_{m_{\ell},\ell}=0
\end{equation}
for $\ell=1,\ldots,N-1$, with $n_{\ell}$ and $m_{\ell}$ being the
row indexes of the two non-zero elements in the $\ell$th column of
$\mathbf{A}_{2}$. Furthermore, the $N$th entry of vector $\mathbf{b}''$
is given by
\begin{align}
\left[\mathbf{b}''\right]_{N} & =\sum_{n=1}^{N}\Re\left\{ \left[\boldsymbol{\alpha}\right]_{n}\right\} \left[\mathbf{b}'\right]_{n}\\
& =\sum_{n=1}^{N}\Re\left\{ \left[\boldsymbol{\alpha}\right]_{n}\right\} \Im\left\{ \left[\boldsymbol{\beta}\right]_{n}\right\} -\Im\left\{ \left[\boldsymbol{\alpha}\right]_{n}\right\} \Re\left\{ \left[\boldsymbol{\beta}\right]_{n}\right\} ,
\end{align}
where
\begin{align}
& \Re\left\{ \left[\boldsymbol{\alpha}\right]_{n}\right\} \Im\left\{ \left[\boldsymbol{\beta}\right]_{n}\right\} -\Im\left\{ \left[\boldsymbol{\alpha}\right]_{n}\right\} \Re\left\{ \left[\boldsymbol{\beta}\right]_{n}\right\} \\
= & jZ_{0}\left(\Im\left\{ \left[\hat{\mathbf{h}}_{RI}^{H}\right]_{n}\right\} ^{2}-\Im\left\{ \left[\mathbf{u}_{IT}\right]_{n}\right\} ^{2}\right.\\
& \left.+\Re\left\{ \left[\hat{\mathbf{h}}_{RI}^{H}\right]_{n}\right\} ^{2}-\Re\left\{ \left[\mathbf{u}_{IT}\right]_{n}\right\} ^{2}\right)\\
= & jZ_{0}\left(\left\vert \left[\hat{\mathbf{h}}_{RI}^{H}\right]_{n}\right\vert ^{2}-\left\vert \left[\mathbf{u}_{IT}\right]_{n}\right\vert ^{2}\right).
\end{align}
Thus, recalling that $\Vert\hat{\mathbf{h}}_{RI}^{H}\Vert^{2}=1$ and $\Vert\mathbf{u}_{IT}\Vert^{2}=1$, we have
\begin{equation}
\left[\mathbf{b}''\right]_{N}=jZ_{0}\left(\left\Vert \hat{\mathbf{h}}_{RI}^{H}\right\Vert ^{2}-\left\Vert \mathbf{u}_{IT}\right\Vert ^{2}\right)=0,
\end{equation}
proving that the $2N$th row of \eqref{eq:Gaussian2} is all-zero.

\bibliographystyle{IEEEtran}
\bibliography{IEEEabrv,main}

\end{document}